\documentclass[11pt]{article}

\newcommand{\ifabs}[2]{#2}
\renewcommand{\ifabs}[2]{#1}

\newcommand{\ifarxiv}[2]{#2}
\renewcommand{\ifarxiv}[2]{#1}

\usepackage{comment}
\usepackage{cite}
\usepackage{fullpage}
\usepackage{color}
\usepackage{epic}
\usepackage{eepic}
\usepackage{epsfig}
\usepackage{xspace}
\usepackage{graphicx}
\usepackage{amsmath}
\usepackage{amssymb}
\usepackage{amsthm}
\usepackage{hyperref}
\usepackage[ruled,vlined]{algorithm2e}

\excludecomment{suppress}

\newcommand{\todo}[1]{\typeout{TODO: \the\inputlineno: #1}\textbf{[[[ #1 ]]]}}

\newcommand{\concept}[1]{\emph{#1}}

\newtheorem{theorem}{Theorem}
\newtheorem{lemma}[theorem]{Lemma}

\newtheorem{definition}[theorem]{Definition}

\newtheorem{claim}[theorem]{Claim}
\newtheorem*{myclaim}{Claim}

\newtheorem*{remark*}{Remark}

\newcommand{\newloglike}[2]{\newcommand{#1}{\mathop{\rm #2}\nolimits}}
\newloglike{\sgn}{sgn}

\newcommand{\nul}[1]{{\it et al.\/}}

\newcommand{\unique}[1]{up-to-{#1} unique}

\begin{document}

\title{Correlation Decay  up to Uniqueness in Spin Systems}
\author{
Liang Li~\thanks{This work was done when this author visited Microsoft Research Asia.}\\
{Shandong University}\\
{\texttt{li.liang@sdu.edu.cn}}
\and
Pinyan Lu\\
{SHUFE}\\
{\texttt{lu.pinyan@mail.shufe.edu.cn}}\and
Yitong Yin\\
Nanjing University\\
\texttt{yinyt@nju.edu.cn}
}
\date{}

\maketitle

\begin{abstract}
We give a complete characterization of the two-state anti-ferromagnetic spin systems which are of strong spatial mixing on general graphs. We show that a two-state anti-ferromagnetic spin system is of strong spatial mixing on all graphs of maximum degree at most $\Delta$ if and only if the system has a unique Gibbs measure on infinite regular trees of degree up to $\Delta$, where $\Delta$ can be either bounded or unbounded. As a consequence, there exists an FPTAS for the partition function of a two-state anti-ferromagnetic spin system on graphs of maximum degree at most $\Delta$ when the uniqueness condition is satisfied on infinite regular trees of degree up to $\Delta$. In particular, an FPTAS exists for arbitrary graphs if the uniqueness is satisfied on all infinite regular trees. This covers as special cases all previous algorithmic results for two-state anti-ferromagnetic systems on general-structure graphs.

Combining with the FPRAS for two-state ferromagnetic spin systems of Jerrum-Sinclair and Goldberg-Jerrum-Paterson, and the very recent hardness results of Sly-Sun and independently of Galanis-\u{S}tefankovi\u{c}-Vigoda, this gives a complete classification, except at the phase transition boundary, of the approximability of all two-state spin systems, on either degree-bounded families of graphs or family of all graphs.

\end{abstract}

\ifabs{ \setcounter{page}{0} \thispagestyle{empty} \vfill
\pagebreak }{
}

\section{Introduction}
Spin systems are well studied in the areas of Statistical Physics, Applied Probability and Computer Science as a general framework to capture the essence of how local interactions and constrains affect the macroscopic properties of  particle systems. A system is usually described by a graph, with each vertex in one of a fixed number of states called spins, and edges specifying the neighborhood relation of the system.

Let $G(V,E)$ be a graph and $q$ be the number of spin states. A configuration of the system is one of the $q^{|V|}$ possible assignments $\sigma: V \rightarrow [q]$. Each configuration $\sigma$ has an energy $H(\sigma)$ as a sum over all edges and vertices, such that the contribution of an edge $(u,v)\in E$ is determined by a symmetric function of the spin states $\sigma(u)$ and $\sigma(v)$, and the contribution of a vertex $v\in V$ is determined by a function of its spin state $\sigma(v)$. The weight of a configuration $\sigma$ is $w(\sigma)=\exp(-\frac{H(\sigma)}{T})$, where $T$ is the temperature. 
We focus on the two-state spin systems.
Up to normalization, a two-state spin system is fully captured by three parameters $(\beta, \gamma, \lambda)$, where $\beta$ and $\gamma$ determine the symmetric function for edge contribution and $\lambda$, also known as the \emph{external field}, determines the function for vertex contribution. 
The Gibbs measure is a natural probability distribution over all configurations such  that the probability of a configuration $\sigma$ is $\frac{w(\sigma)}{Z}$, where the normalizing factor $Z=\sum w(\sigma)$ is called the partition function. The partition function encodes rich information regarding the macroscopic behavior of the spin system. However, for almost all nontrivial settings it is \#P-hard to compute the precise value of partition functions.

One of the most important properties of spin systems is the correlation decay, which says that the correlation between the marginal Gibbs distributions of two vertices decays rapidly with respect to the distance between the two vertices. This property is also called \concept{weak spatial mixing} \cite{Weitz06}. 
Of greater algorithmic significance is the strong spatial mixing, which says that the correlations decay in the presence of arbitrary fixed spins at other vertices. For two-state spin systems, the strong spatial mixing may imply efficient approximation algorithms for the partition function.
It is then an important question to characterize the systems which exhibit strong spatial mixing on arbitrary instances of graphs by the parameters of the systems.

A two-state spin system is ferromagnetic if adjacent vertices favor agreement of spins, and is anti-ferromagnetic if otherwise.  For all two-state ferromagnetic spin systems, the partition functions can be efficiently approximated \cite{JS93,GJP03}. In the anti-ferromagnetic region, the correlation decay plays a central role in the approximability of partition functions. It is believed (see \cite{inapp_MWW09}) that for such models the approximability of the partition function is characterized by the uniqueness of Gibbs measure on infinite regular trees, which is  equivalent to the weak spatial mixing on the infinite regular trees. This condition is called the \emph{uniqueness condition}. 



In a seminal work~\cite{Weitz06}, Weitz shows that for the hardcore model strong spatial mixing is characterized by the uniqueness condition, and in a recent work of Sinclair, Srivastava, and Thurley~\cite{SST} the same characterization is proved for the anti-ferromagnetic Ising model. Both models are important special two-state anti-ferromagnetic spin systems. On the hardness side, it is proved for the hardcore model by Sly~\cite{Sly10} and Galanis \emph{et al.}~\cite{GGSVY11}, and very recently for the general two-state anti-ferromagnetic spin systems by Sly and Sun~\cite{sly2012computational} and independently by Galanis, \u{S}tefankovi\u{c}, and Vigoda~\cite{galanis2012inapproximability} that violating the uniqueness condition implies the inapproximability of partition functions. Two questions remain open for our complete understanding of the correlation decay and computation in two-state spin systems: the characterizations of the strong spatial mixing and approximability of general two-state spin systems.

\subsection{Our results}
We characterize the two-state anti-ferromagnetic spin systems which exhibit strong spatial mixing on general degree-bounded graphs or arbitrary graphs by the uniqueness of Gibbs measure on infinite regular trees. 


\begin{theorem}\label{thm-correlation-decay}
For any finite $\Delta\ge 2$ or $\Delta=\infty$, a two-state anti-ferromagnetic spin system is of strong spatial mixing on graphs of maximum degree at most $\Delta$ if and only if the system exhibits uniqueness on infinite $d$-regular trees for all $d\le\Delta$.
\end{theorem}

Due to Weitz's self-avoiding walk tree construction \cite{Weitz06},
the strong spatial mixing of a two-state spin system on degree-bounded graphs immediately implies an FPTAS for the partition function. Indeed, we show an even stronger notion of correlation decay introduced in a previous work~\cite{LLY}, namely the computationally efficient correlation decay, which gives FPTAS not only for the degree-bounded graphs but also for  arbitrary  graphs with unbounded degrees, when the corresponding uniqueness condition is satisfied. 

\begin{theorem}\label{thm-dichotomy-algorithms}
For any finite $\Delta\ge3$ or $\Delta=\infty$, there exists an FPTAS for the partition function of the  two-state anti-ferromagnetic spin system on graphs of maximum degree at most $\Delta$ if for all $d\le\Delta$ the system parameters lie in the interior of the uniqueness region of infinite $d$-regular tree.
\end{theorem}
In the above two theorems, the $\Delta=\infty$ case represents the graphs of unbounded degree.

Due to a very recent hardness results of Sly and Sun \cite{sly2012computational} for general two-state spin systems and an independent result of Galanis, \u{S}tefankovi\u{c} and Vigoda  \cite{galanis2012inapproximability} for a less general setting, violating the uniqueness condition implies inapproximability of the partition function.

\begin{theorem}[Due to \cite{sly2012computational} and \cite{galanis2012inapproximability}]\label{thm-dichotomy-hardness}
For any finite $\Delta\ge3$ or $\Delta=\infty$, unless $NP=RP$ there does not exist an FPRAS for the partition function of the  two-state anti-ferromagnetic spin system on graphs of maximum degree at most $\Delta$ if for some $d\le\Delta$ the system parameters lie in the interior of the non-uniqueness region of infinite $d$-regular tree.
\end{theorem}

The original theorem in \cite{sly2012computational} holds for $d$-regular graphs with fixed $d$, which immediately implies the hardness for degree-bounded graphs or arbitrary graphs. And the hardness condition in both \cite{sly2012computational} and \cite{galanis2012inapproximability} requires the non-uniqueness on a $d$-regular tree with $d\ge 3$. But the uniqueness on the infinite 2-regular tree (i.e.~the infinite path) always holds for any two-state anti-ferromagnetic spin system, thus the condition in Theorem \ref{thm-dichotomy-hardness} suffices.

For graphs of maximum degree 2 or less, the partition function can be computed exactly in polynomial time. 
Thus Theorem \ref{thm-dichotomy-algorithms} and Theorem \ref{thm-dichotomy-hardness} together with the FPRAS for two-state ferromagnetic spin system \cite{JS93, GJP03} 
give an almost complete (except at the phase transition boundary) classification of the approximability of the partition function of all two-state spin systems, on either degree-bounded families of graphs or family of all graphs.

\subsubsection{Regularity and monotonicity}

In Statistical Physics, the correlation decay is usually studied on regular graphs or even structurally symmetric graphs (e.g.~Bethe lattice). Although for hardness results considering only regular graphs will strengthen the lower bounds, from algorithmic perspective it is more general to consider spin systems on general graphs. The approximation algorithms in~\cite{Weitz06} for the hardcore model and~\cite{SST} for the anti-ferromagnetic Ising model are both for general graphs with bounded maximum degrees. As discussed in~\cite{SST,sly2012computational}, up to translation of parameters the hardcore and Ising models together are complete for general two-state spin systems on $d$-regular graphs with fixed $d$.
However, this does not cover the most general case, namely the general two-state spin system on general graphs, of either bounded or unbounded degree. A fundamental reason for this gap is the non-monotonicity of general spin systems.

The well studied hardcore and anti-ferromagnetic Ising models (along with all two-state spin systems with $\beta,\gamma\le 1$) are both monotone spin systems, in the sense that the uniqueness on infinite $d$-regular tree implies the uniqueness on all infinite regular trees with smaller degrees. This monotonicity does not necessarily hold in general two-state spin systems.

In~\cite{Weitz06}, Weitz established the following implication in the hardcore model:
\begin{myclaim}[Theorem 2.3 in~\cite{Weitz06}]
Strong spatial mixing on a $d$-regular tree implies the strong spatial mixing on all graphs of maximum degree at most $d$.
\end{myclaim}
In~\cite{Weitz06}, Weitz also remarked without proof that this implication holds for all two-state spin systems (indeed this is rigorously proved for anti-ferromagnetic Ising model in~\cite{SST}). With this to be true, devising approximation algorithms for two-state spin systems on degree-bounded graphs is reduced to verifying the strong spatial mixing on $d$-regular trees.
This has been accepted as a fact about the two-state spin systems and has
become a building block for approximation algorithms for such systems (see Theorem 2.8 in~\cite{SST} and the discussions in~\cite{sly2008uniqueness,Sly10}). It was also raised as a conjecture in~\cite{sly2008uniqueness} whether the claim holds for general multi-state spin systems.

As a byproduct of our analysis (see Section~\ref{section-non-monotone}), we find that this well-believed implication between the strong spatial mixing on $d$-regular tree and on graphs of maximum degree at most $d$ holds only for monotone spin systems (including the hardcore and anti-ferromagnetic Ising models). For general two-state spin systems  the worst case for uniqueness as well as strong spatial mixing among all degree-bounded graphs is indeed a regular tree, but may no longer be the one of the highest degree. 
A bright side of this complication is that higher degrees may yield much faster correlation decay, making possible the FPTAS for graphs with unbounded degrees.

These new phenomena suggest that the general two-state spin systems have much richer structure than the well-studied monotone spin systems such as the hardcore and anti-ferromagnetic Ising models. 
The former approach via the strong spatial mixing on $d$-regular tree which succeeds in monotone spin systems on general graphs and general spin systems on regular graphs, meets a barrier when dealing with general spin systems on general graphs.
We give a unified approach to the correlation decay in general two-state spin systems, through the strong spatial mixing on arbitrary trees instead of $d$-regular trees. 
This approach was initiated in our previous work~\cite{LLY} dealing with graphs of unbounded degrees.
In this paper, we devise a unified potential-based analysis which adapts to both the irregularity of the arbitrary tree and the non-monotonicity of general two-state spin system and give tight correlation decay results for all two-state anti-ferromagnetic spin systems on degree-bounded families of graphs and family of all graphs.

\subsubsection{Implications of the main results}


By solving the uniqueness condition we can restate our main results in various threshold forms.
Theorem \ref{thm-dichotomy-algorithms} covers as special cases all previous algorithmic results for two-state anti-ferromagnetic spin systems on general-structure graphs  as well as clears up previously uncovered cases.


\paragraph{In terms of interactions.}
We can fix the external field $\lambda$ and discuss the tractable region of $(\beta,\gamma)$. Since the roles of $\beta$ and $\gamma$ are symmetric, we can further fix one of them and discuss the tractable range of the other. This formulation was used in \cite{GJP03, LLY}.

Our main result can be restated as follows: for any $\Delta$, there is a critical threshold $\gamma_c(\beta,\lambda,\Delta)$ for the uniqueness on infinite regular trees up to degree $\Delta$
such that  if $\gamma_c(\beta,\lambda,\Delta)<\gamma<\frac{1}{\beta}$ there is an FPTAS for graphs of maximum degree at most $\Delta$; and in particular, $\gamma_c=\gamma_c(\beta,\lambda,\infty)>1$ is an absolute constant such that if $\gamma\in(\gamma_c,\frac{1}{\beta})$ there is an FPTAS for arbitrary graphs.

This covers  as special cases all algorithmic results in \cite{GJP03} regarding the anti-ferromagnetic systems and all results in \cite{LLY}, extends the tractable regions in these previous works, and considers the degree-bounded graphs as well.



\paragraph{In terms of external field.}
Motivated by the studies of hardcore and anti-ferromagnetic Ising models, we can fix $(\beta,\gamma)$ and discuss the tractable range of external field. 

Due to the symmetric role of $\beta$ and $\gamma$, we may assume that $\beta\le\gamma$.
Our main results can be restated in specific settings as follows:
\begin{itemize}
\item Hard constraints (when $\beta = 0$): For any $\Delta$, $\lambda_c(\gamma,\Delta)=\min_{1<d< \Delta} \frac{\gamma^{d+1}d^d}{(d-1)^{d+1}}$ is a critical threshold for the uniqueness on infinite regular trees up to degree $\Delta$ 
such that if $\lambda<\lambda_c(\gamma,\Delta)$ there exists an FPTAS for graphs of maximum degree at most  $\Delta$. 

For $\gamma\le1$, the critical threshold equals $\lambda_c(\gamma,\Delta)=\frac{\gamma^{\Delta} (\Delta-1)^{\Delta-1}}{(\Delta-2)^{\Delta}}$. There is no external field $\lambda>0$ satisfying uniqueness on infinite regular trees of unbounded degrees. This is consistent with the hardness result for the hardcore model without degree bound~\cite{IS_DFJ02,Sly10}.
One particularly interesting case is when $\gamma=1$, in which case the model is the hardcore model with  fugacity~$\lambda$, and $\lambda_c(1,\Delta)=\frac{(\Delta-1)^{\Delta-1}}{(\Delta-2)^{\Delta}}$ is the critical threshold. This covers the result of~\cite{Weitz06}.

For $\gamma > 1$, in addition to the results for degree-bounded graphs, there exists an absolute positive constant $\lambda_c(\gamma)=\min_{d>1} \frac{\gamma^{d+1}d^d}{(d-1)^{d+1}}$ which lower bounds $\lambda_c(\gamma,\Delta)$ for all $\Delta$, such that if $\lambda<\lambda_c(\gamma)$ there exists an FPTAS for arbitrary graphs. 

\item 
{Soft constraints (when $\beta>0$): If $\sqrt{\beta \gamma} > \frac{\Delta-2}{\Delta}$ the uniqueness condition holds for any external field $\lambda>0$ and there always exists an FPTAS for graphs of maximum degree at most $\Delta$.

Now suppose $\sqrt{\beta \gamma} \le \frac{\Delta-2}{\Delta}$.
Let $\bar{\Delta}$ be the smallest $d$ satisfying $\sqrt{\beta \gamma} \le \frac{d-1}{d+1}$.
For every integer $\bar{\Delta}\le d<\Delta$, there exist two thresholds $\lambda_1(\beta,\gamma, d)$ and $\lambda_2(\beta,\gamma, d)$ for the uniqueness on the infinite $(d+1)$-regular tree, such that if $\lambda<\lambda_1(\beta,\gamma, d)$ or $\lambda>\lambda_2(\beta,\gamma, d)$ for all integers $\bar{\Delta}\le d< \Delta$, there is an FPTAS for graphs of maximum degree at most $\Delta$.

Furthermore, if $\gamma\le1$, the above uniqueness condition can be simplified as that $\lambda<\lambda_c$ or $\lambda>\bar{\lambda}_c$, where $\lambda_c=\lambda_c(\beta, \gamma, \Delta) \triangleq \min_{\overline{\Delta} \leqslant d < \Delta}\lambda_1(\beta,\gamma,d)$ and $\bar{\lambda}_c=\bar{\lambda}_c(\beta,\gamma,\Delta) \triangleq\max_{\overline{\Delta} \leqslant d<\Delta}\lambda_2(d)$.
In particular, for $\beta=\gamma$, the anti-ferromagnetic Ising model, we have $\bar{\lambda}_c=1/{\lambda}_c$, and the uniqueness holds when $|\log\lambda|>\log\bar{\lambda}_c$.
This covers the result of anti-ferro-Ising model in~\cite{SST}.

For unbounded maximum degree $\Delta$, if $\gamma\le 1$, there is no $\lambda>0$ satisfying the uniqueness on infinite regular trees of all degrees, which  is consistent with the hardness results  in \cite{GJP03};
and if $\gamma>1$, the uniqueness for all infinite regular trees holds when $\lambda<\lambda_1(\beta,\gamma, d)$ or $\lambda>\lambda_2(\beta,\gamma, d)$ for all $d\ge \bar{\Delta}$, under which condition there exists an FPTAS for arbitrary graphs.
}
\end{itemize}

\subsection{Related works }
The approximation of partition functions of spin systems has been extensively studied\cite{JS93, app_JSV04, app_GJ11, app_DJV01, col_Jerrum95, col_Vigoda99, IS_DFJ02, IS_DG00, IS_LV97}. In a seminal work~\cite{JS93}, Jerrum and Sinclair devised a fully polynomial-time randomized approximation scheme (FPRAS) for the ferromagnetic Ising model. Later in~\cite{GJP03}, the FPRAS was extended to all two-state ferromagnetic spin systems by translating the parameters to the ferromagnetic Ising model.
Also in~\cite{GJP03}, Goldberg, Jerrum, and Paterson gave an FPRAS and inapproximability results for two-state anti-ferromagnetic spin systems on arbitrary graphs. 
A gadget based on random regular bipartite graphs was proposed by Dyer, Frieze, and Jerrum in~\cite{IS_DFJ02} and was also analyzed by Mossel, Weitz, and Wormald in~\cite{inapp_MWW09} to study the inapproximability on degree-bounded graphs. 
It is widely believed that the transition of approximability of anti-ferromagnetic spin systems is captured by the phase transition of uniqueness on infinite trees. 
This was raised openly as a conjecture in~\cite{inapp_MWW09}. The conjecture was proved by Sly in~\cite{Sly10} for the hardcore model. This result was improved by Galanis \emph{et al.}~in~\cite{GGSVY11} to a wide range of parameters. Very recently, Sly and Sun~\cite{sly2012computational} proved the hardness of all two-state anti-ferromagnetic spin systems of non-uniqueness on infinite regular trees. A same result for anti-ferromagnetic Ising model without external field was independently proved by Galanis, \u{S}tefankovi\u{c}, and Vigoda in~\cite{galanis2012inapproximability}.

The correlation decay technique developed independently by Weitz~\cite{Weitz06} and Bandyopadhyay and Gamarnik~\cite{BG08} is a powerful tool for devising deterministic fully polynomial-time approximation schemes (FPTAS) for partition functions (other important examples include~\cite{GK07, BGKNT07}). In~\cite{Weitz06}, Weitz introduced the concept of strong spatial mixing and used it to devise FPTAS for the hardcore model up to the uniqueness threshold. The other most important two-state anti-ferromagnetic spin system, the anti-ferromagnetic Ising model, was studied recently by Sinclair, Srivastava, and Thurley in~\cite{SST}, where a more powerful message-decay method was introduced to analyze the strong spatial mixing and give FPTAS up to uniqueness threshold. 
A powerful technique was developed by Restrepo \emph{et al.}~in \cite{RSTVY11} which makes use of the specific structure of graphs for strong spatial mixing.
A broader tractable region than the region of uniqueness is achieved on grid lattice by exploiting the structure of the graph.
In a previous work~\cite{LLY}, we gave an FPTAS for two-state anti-ferromagnetic spin systems without external field on arbitrary graphs with unbounded degrees, up to a continuous relaxation of the uniqueness threshold. The approach used in the current paper was initiated in~\cite{LLY}, however the analysis in~\cite{LLY} cannot separate the uniqueness up to certain degree, thus fails in dealing with degree-bounded families of graphs.

\section{Definitions and preliminaries}\label{section-definition}
A two-state spin system is described by a graph $G=(V,E)$. A \concept{configuration} of the system is one of the $2^{|V|}$ possible
assignments $\sigma: V\rightarrow \{0, 1\}$ of states to vertices. We also use two colors \emph{blue} and \emph{green} to denote these two states. The \concept{weight} of a configuration can be described as a product of contributions of individual edges and vertices.
Let $A=\begin{bmatrix} A_{0,0} & A_{0,1} \\ A_{1,0} & A_{1,1} \end{bmatrix}$ and $b=(b_0,b_1)$. The weight of a configuration $\sigma: V\rightarrow \{0, 1\}$ is given by
\begin{align*}
w(\sigma)=\prod_{(u,v) \in E} A_{\sigma(u), \sigma(v)}\prod_{v\in V} b_{\sigma(v)}.
\end{align*}
The \concept{Gibbs measure} is a probability distribution over all configurations defined by $\rho(\sigma)=\frac{w(\sigma)}{Z(G) }$.
The normalization factor $Z(G)=\sum_{\sigma}w(\sigma)$
is called the \concept{partition function}.

We can normalize the contributions of a $\{\mathtt{blue}, \mathtt{green}\}$ edge and of a green vertex to be 1. So that $A=\begin{bmatrix} \beta & 1 \\ 1 & \gamma \end{bmatrix}$ for some $\beta, \gamma \ge 0$, and $b=(b_0,b_1)=(\lambda,1)$ for some $\lambda>0$. Since the roles of blue and green are symmetric, we can assume that $\beta\le \gamma$ without loss of generality.
The three parameters $(\beta,\gamma,\lambda)$ with $0\le\beta\le\gamma$ and $\lambda>0$ completely specify a two-state spin system. 
A two-state spin system with $\beta=\gamma$ is an \concept{Ising model} and a two-state spin system with $\beta=0, \gamma=1$ or symmetrically $\beta=1,\gamma=0$ is a \concept{hardcore model}.

A two-state spin system is called \concept{anti-ferromagnetic} if adjacent vertices favor disagreeing spins, i.e.~if $\beta\gamma<1$. Without loss of generality, we focus on the cases that $\beta\le\gamma$.
\begin{definition}
$(\beta,\gamma,\lambda)$ is \concept{anti-ferromagnetic} if $0\le\beta\le\gamma$, $\gamma>0$, $\beta\gamma<1$, and $\lambda>0$.
\end{definition}
By the symmetry of $\beta$ and $\gamma$ and the triviality of the case $\beta=\gamma=0$, this definition is complete for all nontrivial two-state anti-ferromagnetic systems

\subsection{Correlation decay}
The Gibbs measure defines a marginal distribution of state for each vertex.
Let $p_v$ denote the probability of vertex $v$ to be colored blue. 
Let $\sigma_{\Lambda}\in\{0,1\}^\Lambda$ be a configuration of $\Lambda\subset V$.  We call vertices $v\in\Lambda$ \concept{fixed} vertices, and $v\not\in\Lambda$ \concept{free} vertices. We use $p_v^{\sigma_{\Lambda}}$ to denote the marginal probability of $v$ to be colored blue conditioning on the configuration of $\Lambda$ being fixed as $\sigma_{\Lambda}$.

\begin{definition}\label{definition-correlation-decay}
A spin system on a family of graphs is said to be of \concept{strong spatial mixing} if for any graph $G=(V,E)$ in the family, any $v\in V,\Lambda\subset V$ and $\sigma_{\Lambda},\tau_{\Lambda}\in\{0,1\}^{\Lambda}$,
\begin{align*}
|p_v^{\sigma_{\Lambda}}-p_v^{\tau_{\Lambda}}|\leq \exp(-\Omega(\mathrm{dist}(v,S))),
\end{align*}
where $S\subset\Lambda$ is the subset on which $\sigma_{\Lambda}$ and $\tau_{\Lambda}$ differ, and $\mathrm{dist}(v,S)$ is the shortest distance from $v$ to any vertex in $S$.
\end{definition}

The \concept{weak spatial mixing} can be defined by measuring the decay with respect to $\mathrm{dist}(v,\Lambda)$ instead of $\mathrm{dist}(v,S)$. The spatial mixing property is also called correlation decay in Statistical Physics.

Let $T$ be a tree rooted by $v$. Define $R_T^{\sigma_\Lambda}=p_v^{\sigma_\Lambda}/(1-p_v^{\sigma_\Lambda})$ to be the ratio between the probabilities that the root $v$ is blue and green, when imposing the condition $\sigma_\Lambda$ (with the convention that $R_T^{\sigma_\Lambda}=\infty$ when $p_v^{\sigma_\Lambda}=1$). Suppose that $v$ has $d$ children and $T_i$ is the subtree rooted by the $i$-th child. Due to the independence of subtrees, we have an easy recursion for calculating $R_T^{\sigma_\Lambda}$:
\begin{align}
R^{\sigma_\Lambda}_T
&=
\lambda\prod_{i=1}^d\frac{\beta R_{T_i}^{\sigma_\Lambda}+1}{R_{T_i}^{\sigma_\Lambda}+\gamma}. \label{eq:recursion}
\end{align}
Let $G(V,E)$ be a graph. Similarly define that $R_{G,v}^{\sigma_\Lambda}=p_v^{\sigma_\Lambda}/(1-p_v^{\sigma_\Lambda})$. In contrast to the case of tree, there is no easy recursion for calculating $R_{G,v}^{\sigma_\Lambda}$ for a general graph $G$ because of the dependencies caused by cycles. In \cite{Weitz06}, a construction called the self-avoiding walk (SAW) tree was introduced which reduces the computing of marginal distribution in a general graph to that in a tree. Specifically, given a graph $G(V,E)$ and a vertex $v\in V$. The SAW tree $T_{\text{SAW}}(G,v)$ is a tree rooted at $v$ with a new vertex set $V_{\text{SAW}}$ (which effectively enumerates all paths originating from $v$ in $G$ and may include fixed leaves). Moreover, any vertex sets $S\subset \Lambda\subset V$ are mapped to respective $S_{\text{SAW}}\subset\Lambda_{\text{SAW}}\subset V_{\text{SAW}}$ and any configuration $\sigma_\Lambda\in\{0,1\}^\Lambda$ is mapped to a $\sigma_{\Lambda_{\text{SAW}}}\in\{0,1\}^{\Lambda_{\text{SAW}}}$. We abuse the notation and write $S=S_{\text{SAW}}$ and $\sigma_\Lambda=\sigma_{\Lambda_{\text{SAW}}}$ if no ambiguity is caused. Given a graph $G(V,E)$, $v\in V$ and $S\subset V$, let $\mathrm{dist}_G(v,S)$ be the shortest distance in $G$ from $v$ to any vertex in $S$.

\begin{theorem}[Theorem 3.1 of Weitz~\cite{Weitz06}]\label{theorem-T-saw}
Let $G(V,E)$ be a graph, $v\in V$, $\sigma_\Lambda\in\{0,1\}^\Lambda$ a configuration on $\Lambda\subset V$, and $S\subset V$. Let $T=T_{\mathrm{SAW}}(G,v)$. It holds that the maximum degree of $T$ equals the maximum degree of $G$, $\mathrm{dist}_G(v,S)=\mathrm{dist}_T(v,S)$, and $R_{G,v}^{\sigma_\Lambda}=R_T^{\sigma_\Lambda}$.
Moreover, any neighborhood of $v$ in $T$ can be constructed in time proportional to the size of the neighborhood. 
\end{theorem}

\subsection{The uniqueness condition}
We consider the uniqueness of Gibbs measure on the infinite $(d+1)$-regular trees $\widehat{\mathbb{T}}^d$, in which the recursion is given by $f_d(x)\triangleq\lambda\left(\frac{\beta x+1}{x+\gamma}\right)^d$ due to the symmetric structure of $\widehat{\mathbb{T}}^d$.

Let $\hat{x}_d$ be the positive fixed point of $f_d(x)$, that is, $\hat{x}_d=f(\hat{x}_d)$. It is known \cite{kelly1985stochastic,MSW07} that the two-state anti-ferromagnetic spin system on $\widehat{\mathbb{T}}^d$ undergoes a phase transition at $|f_d'(\hat{x}_d)|=1$ with uniqueness when $|f_d'(\hat{x}_d)|\le 1$. This motivates the following definition.

\begin{definition}\label{definition-uniqueness}
Let $(\beta,\gamma,\lambda)$ be anti-ferromagnetic. Let $\hat{x}_d$ be the positive fixed point of function $f_d(x)=\lambda\left(\frac{\beta x+1}{x+\gamma}\right)^d$. We say that $(\beta,\gamma,\lambda)$ is \concept{\unique{$\Delta$}}, if for all integers $1\le d< \Delta$,
\begin{eqnarray*}
|f_d'(\hat{x}_d)|
=\frac{\lambda d(1-\beta\gamma)(\beta\hat{x}_d+1)^{d-1}}{(\hat{x}_d+\gamma)^{d+1}}
=\frac{d(1-\beta\gamma)\hat{x}_d}{(\beta\hat{x}_d+1)(\hat{x}_d+\gamma)}
<1.
\end{eqnarray*}
In particular, $(\beta,\gamma,\lambda)$ is \concept{universally unique} if it is \concept{\unique{$\infty$}}.
\end{definition}

Being \unique{$\Delta$} is equivalent to that the system is of weak spatial mixing on infinite regular trees up to degree $\Delta$. The uniqueness condition can be described in various threshold forms, which are given in Appendix \ref{app-threshold}.

The uniqueness is defined by requiring that $|f_d'(\hat{x}_d)|<1$. The following lemma 
states that $|f_d'(\hat{x}_d)|$ is bounded by an absolute constant as long as the uniqueness condition holds.

\begin{lemma}\label{lemma-absolute-uniqueness}
Let $(\beta,\gamma,\lambda)$ be anti-ferromagnetic.
If $(\beta,\gamma,\lambda)$ is \unique{$\Delta$} then there exists an absolute constant $c<1$ which depends only on $\beta$, $\gamma$, $\lambda$ and $\Delta$, such that $|f_d'(\hat{x}_d)|\le c$ for all $1\le d<\Delta$.
\end{lemma}
\begin{proof}
The lemma holds trivially for finite $\Delta$. It then remains to show that in case of universal uniqueness, $|f_d'(\hat{x}_d)|$ cannot be arbitrarily close to 1 as $d$ grows to infinity.
If $(\beta,\gamma,\lambda)$ is universally unique, due to Lemma \ref{lemma-threshold}.\ref{lemma-threshold-infty},  we must have $\gamma>1$.
For  anti-ferromagnetic $(\beta,\gamma,\lambda)$, $\beta\le\frac{1}{\gamma}$, thus the fixed point $\hat{x}_d=\lambda\left(\frac{\beta\hat{x}_d+1}{\hat{x}_d+\gamma}\right)^d\le\frac{\lambda}{\gamma^d}$, therefore $|f_d'(\hat{x}_d)|=\frac{d(1-\beta\gamma)\hat{x}_d}{(\beta\hat{x}_d+1)(\hat{x}_d+\gamma)}\le\frac{d\lambda}{\gamma^d}$. The lemma follows.
\end{proof}

\section{The strong spatial mixing on general graphs}\label{section-correlation-decay}
In this section we prove Theorem \ref{thm-correlation-decay}. The necessity of the uniqueness condition is trivial since strong spatial mixing on general graphs implies weak spatial mixing on regular trees. It then remains to prove the following theorem.

\begin{theorem}\label{theorem-correlation-decay}
Let  $(\beta,\gamma,\lambda)$ be anti-ferromagnetic. For any finite $\Delta\ge 2$ or $\Delta=\infty$,
if $(\beta,\gamma,\lambda)$ is \unique{$\Delta$}, then the two-state spin system of parameters $(\beta,\gamma,\lambda)$ is of strong spatial mixing on graphs of maximum degree at most $\Delta$.
\end{theorem}

Our approach is to prove the strong spatial mixing on arbitrary trees of maximum degree at most $\Delta$, which by the self-avoiding walk tree construction implies the theorem. Note that unlike in~\cite{Weitz06} and~\cite{SST}, we analyze the decay on arbitrary tree instead of regular tree. This is because for general two-state anti-ferromagnetic spin systems the worst case for strong spatial mixing among all graphs of maximum degree at most $d$ may no longer be the $d$-regular tree. 
We will explain this in detail in Section \ref{section-non-monotone}.

Given any graph $G(V,E)$ of maximum degree at most $\Delta$, any configuration $\sigma_\Lambda\in\{0,1\}^\Lambda$ on $\Lambda\subset V$ and any $S\subset\Lambda$, fixing an arbitrary vertex $v\in V$, by Theorem \ref{theorem-T-saw}, a self-avoiding walk tree $T=T_{\text{SAW}}(G,v)$ can be constructed such that the maximum degree of $T$ is bounded by $\Delta$, 
$\mathrm{dist}_G(v,S)=\mathrm{dist}_T(v,S)$ and $R_{G,v}^{\sigma_\Lambda}=R_T^{\sigma_\Lambda}$. Recall that $R_{G,v}^{\sigma_\Lambda}=p_v^{\sigma_\Lambda}/(1-p_v^{\sigma_\Lambda})$ thus $p_v^{\sigma_\Lambda}=R_{G,v}^{\sigma_\Lambda}/(1+R_{G,v}^{\sigma_\Lambda})$.  For any $\sigma_\Lambda$ and $\tau_\Lambda$, 
\begin{align*}
\left|p_v^{\sigma_\Lambda}-p_v^{\tau_\Lambda}\right|=\left|\frac{R_{G,v}^{\sigma_\Lambda}}{1+R_{G,v}^{\sigma_\Lambda}}-\frac{R_{G,v}^{\tau_\Lambda}}{1+R_{G,v}^{\tau_\Lambda}}\right|\le\left|R_{G,v}^{\sigma_\Lambda}-R_{G,v}^{\tau_\Lambda}\right|=\left|R_T^{\sigma_\Lambda}-R_T^{\tau_\Lambda}\right|.
\end{align*}
This motivates the following definition.
\begin{definition}\label{definition-bounds}
Let $T$ be a tree rooted by vertex $v$, $\tau_\Lambda\in\{0,1\}^{\Lambda}$ be a configuration on $\Lambda\subset V$ and $S\subseteq\Lambda$ be a vertex set. Define $R_v$ and $\delta_v$ as that $R_v \le R_T^{\tau_\Lambda}\le R_v+\delta_v$ for all $\sigma_\Lambda\in\{0,1\}^{\Lambda}$ which differ from $\tau_\Lambda$ only on $S$.
\end{definition}
It is then sufficient to prove Theorem \ref{theorem-correlation-decay} by constructing such $R_v$ and $\delta_v$ for $T=T_{\text{SAW}}(G,v)$ and showing that $\delta_v\le \exp(-\Omega(\mathrm{dist}(v,S)))$.

Let $T$ be a tree rooted by $v$, who has $d$ children $v_1,\ldots,v_d$, and $T_i$ be the subtree rooted by $v_i$. It holds that
\begin{align*}
R^{\sigma_\Lambda}_T
&=
f\left(R_{T_1}^{\sigma_\Lambda},\ldots,R_{T_d}^{\sigma_\Lambda}\right)
\triangleq
\lambda\prod_{i=1}^d\frac{\beta R_{T_i}^{\sigma_\Lambda}+1}{R_{T_i}^{\sigma_\Lambda}+\gamma}.
\end{align*}
The lower and upper bounds $R_v$ and $R_v+\delta_v$ can be recursively constructed as follows. The base cases are: (1) $v\in S$, in which case $R_v=0$ and $\delta_v=\infty$; and (2) $v\in\Lambda\setminus S$, i.e.~$v$ is fixed to be the same color in all $\sigma_\Lambda$, in which case $\delta_v=0$ and $R_v=\infty$ (or $R_v=0$) if $v$ is fixed to be blue (or green).
For $v\not\in\Lambda$, since $f(R_1,\ldots,R_d)$ is monotonically decreasing for anti-ferromagnetic $(\beta,\gamma,\lambda)$,
\begin{align*}
R_v
&=
f(R_{v_1}+\delta_{v_1},...,R_{v_d}+\delta_{v_d}),\\
R_v+\delta_v
&=
f(R_{v_1},...,R_{v_d}),
\end{align*}
where $R_{v_i}$ and $R_{v_i}+\delta_{v_i}$ are the corresponding lower and upper bounds for $R_{T_i}^{\sigma_\Lambda}$, $1\le i\le d$. In particular, when $d=0$ the empty product equals 1 by convention, thus $R_v=R_v+\delta_v=\lambda$, which is consistent with the case that $v$ is a free vertex having no children.

By the monotonicity of $f(R_1,\ldots,R_d)$, it is easy to check that the $R_v$ and $R_v+\delta_v$ constructed above satisfy the requirement of Definition \ref{definition-bounds}.
Our goal is to show that $\delta_v$ decays exponentially in depth of recursion when the uniqueness holds.
A straightforward approach is trying to prove that $\delta$ contracts at a constant rate in each recursion step. But this does not have to be true to guarantee the exponential decay. Indeed there are cases that the error does not decay in single steps but decay in a long run. To overcome this, we use a \concept{potential} $\Phi$ to amortize the contraction and show that $\delta\cdot\Phi$ contracts at a constant rate.
%
We choose the potential function to be
\begin{align*}
\Phi(R)=\frac{1}{\sqrt{R(\beta R + 1)(R + \gamma)}}.
\end{align*}
We are analyzing the decay on an arbitrary tree with irregular degrees. In order to adapt this irregularity, the potential function cannot have $d$ as an input, but only caries the information regarding the distribution at the current vertex, yet it has to be able to provide correct compensation to the step-wise decay at any state of $R$ and for all spin systems satisfying sufficient uniqueness. 
A heuristic procedure which leads us to this good potential function is discussed in Appendix \ref{section-potential}. 

Let $\varphi(R)$ be a monotone function satisfying that $\varphi'(R)=\Phi(R)$. We define that
\begin{align*}
\epsilon_v
\triangleq\varphi(R_v+\delta_v)-\varphi(R_v),
\end{align*}
and accordingly, $\epsilon_{v_i}\triangleq\varphi(R_{v_i}+\delta_{v_i})-\varphi(R_{v_i})$, $1\le i\le d$.

We define a function $\alpha(d; x_1,...,x_d)$ as follows:
\begin{align*}
\alpha(d;x_1,...,x_d)
&\triangleq
\frac{(1-\beta\gamma)\left(\lambda\prod^d_{i=1}\frac{\beta x_i+1}{x_i+\gamma}\right)^{\frac{1}{2}}}{\left(\beta\lambda\prod^d_{i=1}\frac{\beta x_i+1}{x_i+\gamma}+1\right)^{\frac{1}{2}}\left(\lambda\prod^d_{i=1}\frac{\beta x_i+1}{x_i+\gamma}+\gamma\right)^{\frac{1}{2}}}\cdot\sum^d_{i=1}\frac{x_i^{\frac{1}{2}}}{(\beta x_i+1)^{\frac{1}{2}}(x_i+\gamma)^{\frac{1}{2}}}.
\end{align*}

The following lemma is obtained from applying the Mean Value Theorem.  Similar routines were previously used in~\cite{LLY,RSTVY11}. The proof of the lemma is in Appendix \ref{app-MVT}.

\begin{lemma}\label{lemma-epsilon}
The followings hold for $\epsilon_v$.
\begin{enumerate}
\item \label{lemma-epsilon-delta}{\em (relation to $\delta_v$)}
$\epsilon_v=\delta_v\cdot\Phi(\widetilde{R})$ for some $\widetilde{R}\in[R_v, R_v+\delta_v]$.

\item \label{lemma-epsilon-bound}{\em (absolute bound)}
Assuming that $\gamma>1$ or the maximum degree of $T$ is bounded by a constant,
if $v\not\in\Lambda$ then $R_v+\delta_v =O(1)$ and
if $v_i\not\in S$ for all $1\le i\le d$ then $\epsilon_v =O(1)$.

\item \label{lemma-epsilon-contract}{\em (stepwise contraction)}
There exist $\widetilde{R_i}\in[R_{v_i}, R_{v_i}+\delta_{v_i}]$, $1\le i\le d$, such that
\begin{align*}
\epsilon_v\le\alpha(d; \widetilde{R}_1,\ldots,\widetilde{R}_d)\cdot\max_{1\le i\le d}\{\epsilon_{v_i}\}.
\end{align*}
\end{enumerate}
\end{lemma}

\newcommand{\ProofLemmaEpsilon}{
\begin{enumerate}
\item Due to the Mean Value Theorem, there exists an $\widetilde{R}\in[R_v, R_v+\delta_v]$ such that
\begin{align*}
\epsilon_v
&=
\varphi(R_v+\delta_v)-\varphi(R_v)
=\varphi'(\widetilde{R})\cdot\delta_v
=\Phi(\widetilde{R})\cdot\delta_v.
\end{align*}
\item
Suppose that each vertex has at most $k$ children. Then due to the assumption either $k$ is bounded or $\gamma>1$.

If $v\not\in\Lambda$, then $\delta_v\le R_v+\delta_v= f(R_{v_1},\ldots,R_{v_d})\le f(0,\ldots,0)=\frac{\lambda}{\gamma^d}$, where $0\le d\le k$.  If $\gamma>1$, then $R_v+\delta_v\le \lambda=O(1)$; and if $k$ is finite, then $R_v+\delta_v\le\max\{\frac{\lambda}{\gamma^k},\lambda\}=O(1)$.

Due to the Mean Value Theorem, there exist $\widetilde{R_i}\in[R_{v_i}, R_{v_i}+\delta_{v_i}]$, $1\le i\le d$, such that
\begin{align*}
\epsilon_v
&=
\varphi\left(f(R_{v_1},\ldots,R_{v_d})\right)-\varphi\left(f(R_{v_1}+\delta_{v_1},\ldots,R_{v_d}+\delta_{v_d})\right) \\
&=
-\nabla \varphi\left(f(\widetilde{R_1},\ldots,\widetilde{R_d})\right)\cdot(\delta_{v_1},\ldots,\delta_{v_d})\\
&=
\frac{(1-\beta \gamma) \cdot\left(f(\widetilde{R_1},\ldots,\widetilde{R_d})\right)^{\frac{1}{2}}}{\left(\beta f(\widetilde{R_1},\ldots,\widetilde{R_d})+1\right)^{\frac{1}{2}}\left(f(\widetilde{R_1},\ldots,\widetilde{R_d})+\gamma\right)^{\frac{1}{2}}}
  \cdot \sum_{i=1}^d  \frac{\delta_{v_i}}{ (\beta \widetilde{R_i} +1)(\widetilde{R_i} +\gamma)} \\
&\le
\sqrt{\frac{f(0,\ldots,0)}{\gamma}}\cdot\sum_{i=1}^d\frac{\delta_{v_i}}{\gamma}.
\end{align*}
If $v_i\not\in S$ for all $v_i$, then due to the above argument, $\delta_{v_i}\le \frac{\lambda}{\gamma^{d_i}}$ for free $v_i$, where $0\le d_i\le k$ is the number of children of $v_i$; and trivially $\delta_{v_i}=0$ for fixed $v_i$. Therefore, $\epsilon_v\le\sqrt{\frac{\lambda}{\gamma^{d+1}}}\sum_{i=1}^d\frac{\lambda}{\gamma^{d_i+1}}\le\lambda^{\frac{3}{2}}d\gamma^{-\frac{d+3}{2}}\cdot\max\left\{\frac{1}{\gamma^k},1\right\}$. If $\gamma>1$, then it is easy to see that  $d\gamma^{-\frac{d+3}{2}}$ is bounded by a constant, thus it holds that $\epsilon_v=O(1)$; if $k$ is bounded, then $d\le k$ is also bounded, thus clearly $\epsilon_v=O(1)$.

\item
We then analyze the stepwise contraction of $\epsilon_v$. Define that $y_v=\varphi(R_v)$ and accordingly $y_{v_i}=\varphi(R_{v_i})$, $1\le i\le d$. Then $y_v+\epsilon_v=\varphi(R_v+\delta_v)$ and $y_{v_i}+\epsilon_{v_i}=\varphi(R_{v_i}+\delta_{v_i})$, $1\le i\le d$.  
We have
\begin{eqnarray*}
y_v&=&\varphi(f(\varphi^{-1}(y_{v_1}+\epsilon_{v_1}),...,\varphi^{-1}(y_{v_d}+\epsilon_{v_d}))),\\
y_v+\epsilon_v&=&\varphi(f(\varphi^{-1}(y_{v_1}),...,\varphi^{-1}(y_{v_d}))).
\end{eqnarray*}
Apply the Mean Value Theorem. There exist $\widetilde{y_i}\in[y_{v_i}, y_{v_i}+\epsilon_{v_i}]$ and corresponding $\widetilde{R_i}\in[R_{v_i}, R_{v_i}+\delta_{v_i}]$ satisfying $\widetilde{y_i}=\varphi(\widetilde{R_i}), 1\le i\le d$, such that
\begin{align*}
\epsilon_v
&=
\varphi(f(\varphi^{-1}(y_{v_1}),...,\varphi^{-1}(y_{v_d})))-\varphi(f(\Phi^{-1}(y_{v_1}+\epsilon_{v_1}),...,\varphi^{-1}(y_{v_d}+\epsilon_{v_d})))\\
&=
-\nabla\varphi(f(\varphi^{-1}(\widetilde{y_1}),...,\varphi^{-1}(\widetilde{y_d})))\cdot(\epsilon_{v_1},...,\epsilon_{v_d})\\
&=
-\Phi(f(\widetilde{R_1},...,\widetilde{R_d}))\cdot\sum_{i=1}^d\frac{\partial f}{\partial R_i}\frac{1}{\Phi(\widetilde{R_i})}\cdot \epsilon_i\\
&=
\frac{(1-\beta\gamma)\left(\lambda\prod^d_{i=1}\frac{\beta\widetilde{R_i}+1}{\widetilde{R_i}+\gamma}\right)^{\frac{1}{2}}}{\left(\beta\lambda\prod^d_{i=1}\frac{\beta\widetilde{R_i}+1}{\widetilde{R_i}+\gamma}+1\right)^{\frac{1}{2}}\left(\lambda\prod^d_{i=1}\frac{\beta\widetilde{R_i}+1}{\widetilde{R_i}+\gamma}+\gamma\right)^{\frac{1}{2}}}\cdot\sum^d_{i=1}\frac{\epsilon_{v_i}\widetilde{R_i}^{\frac{1}{2}}}{(\beta\widetilde{R_i}+1)^{\frac{1}{2}}(\widetilde{R_i}+\gamma)^{\frac{1}{2}}}\\
&\le
\max_{1\le i\le d}\{\epsilon_{v_i}\}\cdot\frac{(1-\beta\gamma)\left(\lambda\prod^d_{i=1}\frac{\beta\widetilde{R_i}+1}{\widetilde{R_i}+\gamma}\right)^{\frac{1}{2}}}{\left(\beta\lambda\prod^d_{i=1}\frac{\beta\widetilde{R_i}+1}{\widetilde{R_i}+\gamma}+1\right)^{\frac{1}{2}}\left(\lambda\prod^d_{i=1}\frac{\beta\widetilde{R_i}+1}{\widetilde{R_i}+\gamma}+\gamma\right)^{\frac{1}{2}}}\cdot\sum^d_{i=1}\frac{\widetilde{R_i}^{\frac{1}{2}}}{(\beta\widetilde{R_i}+1)^{\frac{1}{2}}(\widetilde{R_i}+\gamma)^{\frac{1}{2}}}\\
&=
\alpha(d; \widetilde{R}_1,\ldots,\widetilde{R}_d)\cdot\max_{1\le i\le d}\{\epsilon_{v_i}\}.
\end{align*}
\end{enumerate}
}


To prove the strong spatial mixing,
we first relate $\delta_v$ to $\epsilon_v$ by Item~\ref{lemma-epsilon-delta} of Lemma~\ref{lemma-epsilon}, and then apply induction on the depth in $T$, with Item~\ref{lemma-epsilon-bound} of Lemma~\ref{lemma-epsilon} as basis and Item~\ref{lemma-epsilon-contract} of Lemma~\ref{lemma-epsilon} as induction step. We then need to bound the contraction rate $\alpha(d;x_1,...,x_d)$.
Note that $\frac{z}{(\beta z+1)(z+\gamma)}\le\left(1+\sqrt{\beta\gamma}\right)^{-2}\le 1$ for $z\in[0,\infty)$, thus it holds unconditionally for all $x_i\in[0,\infty)$, $1\le i\le d$, that
\begin{align}
\alpha(d;x_1,...,x_d)
&\le d, \label{eq:alpha-unconditional}\\
\alpha(d;x_1,...,x_d)
&\le d\cdot\sqrt{\frac{\lambda}{\gamma^{d+1}}}. \label{eq:alpha-unbounded}
\end{align}
With the uniqueness,  the following much tighter contraction bound can be proved.
\begin{lemma}\label{lemma-asymmetric-alpha-bound}
Let  $(\beta,\gamma,\lambda)$ be anti-ferromagnetic.
If $(\beta,\gamma,\lambda)$ is \unique{$\Delta$}, then there exists a constant $\alpha<1$ such that for any integer $1\le d<\Delta$ and any $x_1,...,x_d\in[0,+\infty)$, $1\le i\le d$, it holds that $\alpha(d;x_1,\ldots,x_d)\le\alpha$.
\end{lemma}

This lemma is the technical core of our analysis. It crucially relies on the choice of potential function.
Before delving into the formal proof of Lemma \ref{lemma-asymmetric-alpha-bound}, we note that Theorem \ref{theorem-correlation-decay} can be implied by this lemma.

\subsection*{Proof of Theorem \ref{theorem-correlation-decay}.}
Let $T=T_{\text{SAW}}(G,v)$ for a $G$ whose maximum degree is at most $\Delta$. Then the maximum degree of $T$ is at most $\Delta$, thus the root $v$ has at most $\Delta$ children in $T$, and every other vertex in $T$ has less than $\Delta$ children. We recursively construct $R_u$, $\delta_u$ and $\epsilon_u$ for every subtree in $T$.

Let $t=\mathrm{dist}(v,S)$. By repeatedly applying Item \ref{lemma-epsilon-contract} of Lemma \ref{lemma-epsilon}, without loss of generality, we have a path $u_1u_2\cdots u_{t-2}$ in $T$ with $u_1=v$ such that $\epsilon_{u_j}\le\alpha(d_j;x_1,\ldots,x_{d_j})\cdot\epsilon_{u_{j+1}}$ for  $j=1,2,\ldots,t-3$, where $d_j$ is the number of children of $u_j$ and $x_i\in[0,\infty)$, $1\le i\le d_j$.

Note that $d_1\le\Delta$, and $d_j<\Delta$ for all other $j$. Assume that $(\beta,\gamma,\lambda)$ is \unique{$\Delta$}. If $\Delta$ is bounded, then by Lemma \ref{lemma-asymmetric-alpha-bound} there exists a constant $\alpha<1$, such that
$\epsilon_{u_j}\le \alpha\cdot \epsilon_{u_{j+1}}$ for $2\le j\le t-3$, and $\epsilon_{v}\le d_1\cdot\epsilon_{u_2}\le \Delta\cdot\epsilon_2$ due to \eqref{eq:alpha-unconditional}; and if $\Delta=\infty$, then by Lemma \ref{lemma-asymmetric-alpha-bound}, $\epsilon_{u_j}\le \alpha\cdot \epsilon_{u_{j+1}}$ for all $1\le j\le t-3$.
In both cases we have $\epsilon_v=O(\alpha^t\cdot\epsilon_{u_{t-2}})$.

Due to Item \ref{lemma-epsilon-delta} of Lemma \ref{lemma-epsilon},
$\delta_v=\frac{\epsilon_v}{\Phi(\widetilde{R})}=O\left(\frac{1}{\Phi(\widetilde{R})}\cdot\alpha^t\epsilon_{u_{t-2}}\right)$ for some $\widetilde{R}\in[R_v,R_v+\delta_v]$. We then bound $\Phi(\widetilde{R})$ and $\epsilon_{u_{t-2}}$.
Due to Item \ref{lemma-threshold-infty} of Lemma \ref{lemma-threshold}, the fact that $(\beta,\gamma,\lambda)$ is \unique{$\Delta$} implies that either $\Delta$ is bounded or $\gamma>1$. Note that $v$ must be free or the theorem is trivial to prove, and none of $u_{t-2}$'s children is in $S$ because $\mathrm{dist}(v,S)=t$. Thus by Item \ref{lemma-epsilon-bound} of Lemma \ref{lemma-epsilon}, $\epsilon_{u_{t-2}}=O(1)$ and $\widetilde{R}\le R_v+\delta_v=O(1)$, which implies that  $\Phi(\widetilde{R})=\frac{1}{\sqrt{\widetilde{R}(\beta \widetilde{R}+1)(\widetilde{R}+\gamma)}}=\Omega(1)$. 

In conclusion, if $(\beta,\gamma,\lambda)$ is \unique{$\Delta$}, there exists a constant $\alpha<1$, such that $\delta_v=O\left(\alpha^{t}\right)$.
As discussed in the beginning of this section, this proves Theorem \ref{theorem-correlation-decay}.

\subsection*{Proof of Lemma \ref{lemma-asymmetric-alpha-bound}.}
The rest of this section is devoted to the proof of Lemma \ref{lemma-asymmetric-alpha-bound}. Given that $(\beta,\gamma,\lambda)$ is \unique{$\Delta$}, there exists an absolute constant $\alpha<1$ such that $\alpha(d;x_1,\ldots,x_d)\le \alpha$ for any $1\le d<\Delta$ and $x_1,\ldots,x_d\ge 0$.

We define the symmetric version of $\alpha(d;x_1,\ldots,x_d)$:
\begin{align*}
\alpha_d(x)
&\triangleq
\alpha(d;\underbrace{x,\ldots,x}_{d})
=
\frac{d(1-\beta\gamma)\left({x}\cdot\lambda\left(\frac{\beta{x}+1}{{x}+\gamma}\right)^d\right)^{\frac{1}{2}}}{(\beta{x}+1)^{\frac{1}{2}}({x}+\gamma)^{\frac{1}{2}}\left(\beta\lambda\left(\frac{\beta{x}+1}{{x}+\gamma}\right)^d+1\right)^{\frac{1}{2}}\left(\lambda\left(\frac{\beta{x}+1}{{x}+\gamma}\right)^d+\gamma\right)^{\frac{1}{2}}}.
\end{align*}
The following lemma shows that the symmetric case dominates the maximum of $\alpha(d; x_1,...,x_d)$ by using the inequalities of Cauchy-Schwarz and arithmetic and geometric means.

\begin{lemma}\label{lemma-asymmetric-less-than-symmetric}
Let  $(\beta,\gamma,\lambda)$ be anti-ferromagnetic.
Then for any integer $d$ and any $x_1,...,x_d\in[0,+\infty)$, there exists an $\bar{x}\in[0,+\infty)$ such that $\alpha(d; x_1,...,x_d)\le \alpha(d,\bar{x})$.
\end{lemma}

\newcommand{\ProofLemmaSymmetry}{
Let $z_i=\frac{\beta x_i+1}{x_i+\gamma}$. Then $z_i\in(\beta,\frac{1}{\gamma}]$ and $x_i= \frac{1-\gamma z_i}{z_i-\beta}$.
Express $\alpha(d;x_1,...,x_d)$ in terms of $z_i$:
\begin{align*}
\alpha(d;x_1,...,x_d)
&=
\frac{\left(\lambda\prod^d_{i=1}z_i\right)^{\frac{1}{2}}}{\left(\beta\lambda\prod^d_{i=1}z_i+1\right)^{\frac{1}{2}}\left(\lambda\prod^d_{i=1}z_i+\gamma\right)^{\frac{1}{2}}} \cdot\sum^d_{i=1}(z_i^{-1}-\gamma)^{\frac{1}{2}}(z_i-\beta)^{\frac{1}{2}}.
\end{align*}
Due to Cauchy-Schwarz inequality,
\begin{align*}
\sum^d_{i=1}(z_i^{-1}-\gamma)^{\frac{1}{2}}(z_i-\beta)^{\frac{1}{2}}
\le
d\left(\frac{1}{d}\sum^d_{i=1}(z_i^{-1}-\gamma)(z_i-\beta)\right)^{\frac{1}{2}}
=
d\left(1+\beta\gamma-\frac{1}{d}\sum^d_{i=1}(z_i\gamma+\beta z_i^{-1})\right)^{\frac{1}{2}}.
\end{align*}
Due to the inequality of arithmetic and geometric means,
\begin{align*}
d\left(1+\beta\gamma-\frac{1}{d}\sum^d_{i=1}(z_i\gamma+\beta z_i^{-1})\right)^{\frac{1}{2}}
&\le
d\left(1+\beta\gamma-\gamma\left(\prod^d_{i=1}z_i\right)^{\frac{1}{d}}-\beta\left(\prod^d_{i=1}z_i\right)^{-\frac{1}{d}} \right)^{\frac{1}{2}}.
\end{align*}
Let $\bar{z}=\left(\prod^d_{i=1}z_i\right)^{\frac{1}{d}}$. Then combining the above calculations,
\begin{align*}
\alpha(d;x_1,...,x_d)
&\le
\frac{(\lambda\bar{z}^d)^{\frac{1}{2}}\cdot d(1+\beta\gamma-\gamma\bar{z}-\beta\bar{z}^{-1})^{\frac{1}{2}}}{(\beta\lambda\bar{z}^d+1)^{\frac{1}{2}}(\lambda\bar{z}^d+\gamma)^{\frac{1}{2}}}
=
d\cdot \sqrt{\frac{\lambda\bar{z}^d(\bar{z}^{-1}-\gamma)(\bar{z}-\beta)}{(\beta\lambda\bar{z}^d+1)(\lambda\bar{z}^d+\gamma)}}.
\end{align*}

Let $\bar{x}$ be such that $\frac{\beta\bar{x}+1}{\bar{x}+\gamma}=\bar{z}$. Then $\bar{x}\in[0,+\infty)$ and by substituting $\frac{\beta\bar{x}+1}{\bar{x}+\gamma}$ for $\bar{z}$, we have
\begin{align*}
\alpha(d; x_1,...,x_d)
&\le
\frac{d(1-\beta\gamma)\left(\bar{x}\cdot\lambda\left(\frac{\beta\bar{x}+1}{\bar{x}+\gamma}\right)^d\right)^{\frac{1}{2}}}{(\beta\bar{x}+1)^{\frac{1}{2}}(\bar{x}+\gamma)^{\frac{1}{2}}\left(\beta\lambda\left(\frac{\beta\bar{x}+1}{\bar{x}+\gamma}\right)^d+1\right)^{\frac{1}{2}}\left(\lambda\left(\frac{\beta\bar{x}+1}{\bar{x}+\gamma}\right)^d+\gamma\right)^{\frac{1}{2}}}
=\alpha_d(\bar{x}).
\end{align*}
}
\begin{proof}
\ProofLemmaSymmetry
\end{proof}

\begin{lemma}\label{lemma-CD-gamma}
Let  $(\beta,\gamma,\lambda)$ be anti-ferromagnetic.
If $(\beta,\gamma,\lambda)$ is \unique{$\Delta$}, then there exists a constant $\alpha<1$ such that for any integer $1\le d<\Delta$, it holds that $\alpha_d(x)\le\alpha$ for all $x\ge 0$.
\end{lemma}
\begin{proof}
Fix $d$ to be any positive integer. We characterize the value of $x$ at which $\alpha_d(x)$ achieves its maximum.
We denote that $f_d(x)=\lambda\left(\frac{\beta x+1}{x+\gamma}\right)^d$. Taking derivative of $\alpha_d(x)$ with respect to $x$, we get that
\begin{align*}
\alpha_d'(x)=d(1-\beta\gamma)\cdot\frac{G'(x)}{2\sqrt{G(x)}},
\end{align*}
where $G(x)
=\frac{xf_d(x)}{(\beta f_d(x)+1)(f_d(x)+\gamma)(\beta x+1)(x+\gamma)}$, whose derivative is
\begin{align*}
G'(x)
&=
\frac{f_d(x)\cdot d(1-\beta\gamma)x}{(\beta f_d(x)+1)(f_d(x)+\gamma)(\beta x+1)^2(x+\gamma)^2}\cdot\left(\frac{\gamma-\beta x^2}{d(1-\beta\gamma)x}-\frac{\gamma-\beta f_d(x)^2}{(\beta f_d(x)+1)(f_d(x)+\gamma)}\right).
\end{align*}
As  $x$ ranges over $[0,\infty)$, the function $\frac{\gamma-\beta x^2}{d(1-\beta\gamma)x}$ is strictly decreasing in $x$ and ranges from $+\infty$ to $-\infty$, and the function $\frac{\gamma-\beta f_d(x)^2}{\left(\beta f_d(x)+1\right)\left(f_d(x)+\gamma\right)}$ is strictly increasing in $x$ and has a bounded range. Thus, the equation
\begin{align}
\frac{\gamma-\beta x^2}{d(1-\beta\gamma)x}
&=
\frac{\gamma-\beta f_d(x)^2}{\left(\beta f_d(x)+1\right)\left(f_d(x)+\gamma\right)}\label{eq:x_d}.
\end{align}
has unique solution in $(0,\infty)$, denoted by $x_d$.
Moreover, it holds that
\begin{align}
G'(x)
\begin{cases}
>0 & \text{if }0\le x<x_d,\\
=0 & \text{if }x=x_d,\\
<0 & \text{if }x>x_d.
\end{cases}\label{eq:sign-G'(x)}
\end{align}
The same also holds for $\alpha'_d(x)$. Thus, for any fixed $d$, $\alpha_d(x)$ achieves its maximum when $x=x_d$.

Therefore,  for all $x\ge 0$,
\begin{align}
\alpha_d(x)
\le
\alpha_d(x_d)
&=
d(1-\beta\gamma)\left(\frac{x_df_d(x_d)}{(\beta{x}_d+1)({x}_d+\gamma)\left(\beta f_d(x_d)+1\right)\left(f_d(x_d)+\gamma\right)}\right)^{\frac{1}{2}}\notag\\
&=
\left(d(1-\beta\gamma)\cdot\frac{(\gamma-\beta x_d^2)}{(\beta{x}_d+1)({x}_d+\gamma)}\cdot\frac{f_d(x_d)}{\left(\gamma-\beta f_d(x_d)^2\right)}\right)^{\frac{1}{2}}\label{eq:alpha-x_d}\\
&\triangleq
\tilde{\alpha}_d(x_d).\notag
\end{align}
Equation \eqref{eq:alpha-x_d} is obtained by substituting 
$\left(\beta f_d(x_d)+1\right)\left(f_d(x_d)+\gamma\right)$ according to \eqref{eq:x_d}.

Let $\hat{x}_d$ be the positive fixed point of $f_d(x)$, that is, $\hat{x}_d=f_d(\hat{x}_d)$.
We then claim that if $(\beta,\gamma,\lambda)$ is \unique{$\Delta$}, then $\tilde{\alpha}_d(x_d)\le\tilde{\alpha}_d(\hat{x}_d)$ for any integral $1\le d<\Delta$,
To see that this claim is sufficient to imply the lemma, note that after substituting $\hat{x}_d=f_d(\hat{x}_d)$, we have $\tilde{\alpha}_d(\hat{x}_d)=\sqrt{\frac{d(1-\beta\gamma)\hat{x}_d}{(\beta\hat{x}_d+1)(\hat{x}_d+\gamma)}}=\sqrt{\left|f_d'(\hat{x}_d)\right|}$.
And due to Lemma \ref{lemma-absolute-uniqueness}, if $(\beta,\gamma,\lambda)$ is \unique{$\Delta$} then there exists a constant $c<1$ such that $|f_d'(\hat{x}_d)|<c$ for all integer $1\le d<\Delta$.

We then prove the claim. Assume that $(\beta,\gamma,\lambda)$ is \unique{$\Delta$} and $1\le d< \Delta$. It is then sufficient to show that $\tilde{\alpha}_d(x)$ is decreasing if $\hat{x}_d\le x_d$ and is increasing if $\hat{x}_d> x_d$.

\noindent Case 1: $\hat{x}_d\le x_d$. Due to \eqref{eq:sign-G'(x)}, $G'(\hat{x}_d)\ge 0$. Note that
\begin{align*}
G'(\hat{x}_d)
=
\frac{d(1-\beta\gamma)(\gamma-\beta\hat{x}_d^2)\hat{x}_d^2}{(\beta\hat{x}_d+1)^3(\hat{x}_d+\gamma)^3}\cdot\left(\frac{1}{d(1-\beta\gamma)\hat{x}_d}-\frac{1}{(\beta\hat{x}_d+1)(\hat{x}_d+\gamma)}\right).
\end{align*}
Due to the uniqueness, $|f_d'(\hat{x}_d)|=\frac{d(1-\beta\gamma)\hat{x}_d}{(\beta\hat{x}_d+1)(\hat{x}_d+\gamma)}<1$, thus $\frac{1}{d(1-\beta\gamma)\hat{x}_d}-\frac{1}{(\beta\hat{x}_d+1)(\hat{x}_d+\gamma)}>0$. Combining with that $G'(x)\ge 0$, we have $\gamma-\beta\hat{x}_d^2\ge 0$. Since the function $f_d(x)$ is monotonically decreasing and $\hat{x}_d$ is its fixed point,
$\gamma-\beta f_d({x}_d)^2
\ge
\gamma-\beta f_d(\hat{x}_d)^2
=
\gamma-\beta\hat{x}_d^2
\ge
0$. Since $x_d$ satisfies \eqref{eq:x_d}, $\gamma-\beta{x}_d^2$ and $\gamma-\beta f_d({x}_d)^2$ must be simultaneously positive or negative, thus it also holds that $\gamma-\beta{x}_d^2\ge 0$. Then 
both $\frac{(\gamma-\beta x^2)}{(\beta{x}+1)({x}+\gamma)}$ and $\frac{f_d(x)}{\left(\gamma-\beta f_d(x)^2\right)}$ are positive and monotonically decreasing in $x\in[\hat{x}_d,x_d]$. Therefore, $\tilde{\alpha}_d(x_d)\le \tilde{\alpha}_d(\hat{x}_d)$.

\noindent Case 2: $\hat{x}_d> x_d$. By the same argument as above, it holds that $\gamma-\beta f_d(\hat{x}_d)^2=\gamma-\beta\hat{x}_d^2<0$, $\gamma-\beta f_d({x}_d)^2<0$, and $\gamma-\beta{x}_d^2<0$. Thus both $\frac{(\gamma-\beta x^2)}{(\beta{x}+1)({x}+\gamma)}$ and $\frac{f_d(x)}{\left(\gamma-\beta f_d(x)^2\right)}$ are negative and monotonically decreasing in $x\in[x_d,\hat{x}_d]$, hence their product is positive and increasing in $x\in[x_d,\hat{x}_d]$.  Therefore, $\tilde{\alpha}_d(x_d)\le \tilde{\alpha}_d(\hat{x}_d)$.
\end{proof}

Lemma \ref{lemma-asymmetric-alpha-bound} is proved by combining Lemma \ref{lemma-asymmetric-less-than-symmetric} and  \ref{lemma-CD-gamma}. This completes the proof of Theorem~\ref{theorem-correlation-decay}.

\subsection*{Strong spatial mixing on regular trees.} 
As a byproduct of our analysis, we prove a strong spatial mixing theorem for regular trees.
When the graphs $G$ itself is a regular tree. All vertices (except the root) has the same arity. And all $d$s (excerpt the one of the root) that appear in the proof are the same and equal that arity. Then the condition that the uniqueness holds on all infinite regular trees of degree up to $\Delta$ can be replaced by the uniqueness on infinite $\Delta$-regular tree. 

\begin{theorem}\label{theorem-regular-decay}
For two-state anti-ferromagnetic spin systems, on any infinite $\Delta$-regular tree the uniqueness implies the strong spatial mixing.
\end{theorem}

The same result can be obtained by combining the same theorem for the hardcore model~\cite{Weitz06} and anti-ferromagnetic Ising model~\cite{SST} and translating the parameters of general two-state anti-ferromagnetic spin systems to these  models (as discussed in\cite{sly2012computational, SST}). However, unlike the hardcore and the anti-ferromagnetic Ising models, for general two-state anti-ferromagnetic spin systems Theorem \ref{theorem-regular-decay} itself is not sufficient to imply the strong spatial mixing on graphs of maximum degree at most $\Delta$. This is discussed in Section \ref{section-non-monotone}.

\section{Algorithmic implications}
In this section we prove Theorem \ref{thm-dichotomy-algorithms}. That is, if an anti-ferromagnetic $(\beta,\gamma,\lambda)$ is \unique{$\Delta$} then there exists an FPTAS for the partition function $Z(G)$ for any graph $G$ of maximum degree at most $\Delta$, and in particular the universal uniqueness implies an FPTAS for arbitrary graph $G$.

It is well-known that $Z(G)$ can be computed from $p_v^{\sigma_\Lambda}$ by the following standard procedure. Let $v_1,\ldots,v_n$ enumerate the vertices in $G$. For $0\le i\le n$, let $\sigma_i$ be the configuration fixing the first $i$ vertices $v_1,\ldots,v_i$ as follows: $\sigma_i(v_j)=\sigma_{i-1}(v_j)$ for $1\le j\le i-1$ and $\sigma_i(v_i)$ is fixed so that $p_{i}\triangleq\Pr[\sigma_i(v_i)\mid \sigma_{i-1}]\ge 1/3$. In particular, $\sigma_n\in\{0,1\}^V$ is a configuration of $V$. It holds for the Gibbs measure of $\sigma_n$ that $\rho(\sigma_n)=p_1 p_2\cdots p_n$ as well as that $\rho(\sigma_n)=\frac{w(\sigma_n)}{Z(G)}$, thus $Z(G)=\frac{w(\sigma_n)}{p_1p_2\cdots p_n}$, where the weight $w(\sigma_n)=\prod_{(u,v) \in E} A_{\sigma_n(u), \sigma_n(v)}\prod_{v\in V} b_{\sigma_n(v)}$ can be computed precisely for any particular $\sigma_n$ in time polynomial in $n$.
Note that $p_i$ equals to either $p_{v_i}^{\sigma_{i-1}}$ or $1-p_{v_i}^{\sigma_{i-1}}$.
Therefore, if $p_v^{\sigma_\Lambda}$ can be approximated within an additive error $\epsilon$ in time polynomial in $n$ and $\frac{1}{\epsilon}$, then the configurations $\sigma_i$ can be efficiently constructed such that all $p_i$ are bounded away from 0, thus the product $p_1p_2\cdots p_n$ can be approximated within a factor of $(1\pm n\epsilon)$ in time polynomial in $n$ and $\frac{1}{\epsilon}$, which implies an FPTAS  for $Z(G)$.

\paragraph{Bounded degree graphs.}
Let $G$ be a graph whose maximum degree is at most $\Delta$ and $v$ be any vertex. A self-avoiding walk tree $T=T_{\text{SAW}}(G,v)$ can be constructed so that $R_{G,v}^{\sigma_\Lambda}=R_T^{\sigma_\Lambda}$. We can use the recursive procedure described in Section \ref{section-correlation-decay} to compute the upper and lower bounds of $R_T^{\sigma_\Lambda}$, with the setting that for all the vertices more than $t$ steps away from the root $v$, the trivial bounds $0\le R_T^{\sigma_\Lambda}\le\infty$ is used. Then the proof of Theorem \ref{theorem-correlation-decay} shows that the recursive procedure returns $R_0$ and $R_1$ such that $R_0\le R_T^{\sigma_\Lambda}\le R_1$, and $R_1-R_0=O(\alpha^t)$ for some constant $\alpha<1$ assuming that $(\beta,\gamma,\lambda)$ is \unique{$\Delta$}. Note that $R_T^{\sigma_\Lambda}=R_{G,v}^{\sigma_\Lambda}=\frac{p_v^{\sigma_\Lambda}}{1-p_v^{\sigma_\Lambda}}$.
Let $p_0=\frac{R_0}{R_0+1}$ and $p_1=\frac{R_1}{R_1+1}$. Then $p_0\le p_v^{\sigma_\Lambda}\le p_1$ and
\begin{align}
p_1-p_0=\frac{R_1}{R_1+1}-\frac{R_0}{R_0+1}\le R_1-R_0=O(\alpha^t). \label{eq:p-bound}
\end{align}
The recursive procedure runs in time $O(\Delta^t)$ since it only needs to construct the first $t$ levels of the self-avoiding walk tree. If $\Delta$ is bounded, by setting $t=\log_{\alpha}\epsilon$, this gives an algorithm which approximates $p_v^{\sigma_\Lambda}$ within an additive error $O(\epsilon)$ in time polynomial in $n$ and $\frac{1}{\epsilon}$, which implies an FPTAS for $Z(G)$.

\paragraph{Arbitrary graphs.}
Let $G$ be an arbitrary graph and $v$ be any vertex. Let $T=T_{\text{SAW}}(G,v)$.
We use the method of \concept{Computationally Efficient Correlation Decay} introduced in~\cite{LLY} to deal with the unbounded degrees. Intuitively, using this method we observe correlation decay in a refined metric instead of graph distance such that in this new metric a neighborhood of moderate size is sufficient to guarantee desirable correlation decay.

Similarly, we use the recursive procedure described in Section \ref{section-correlation-decay} to compute the upper and lower bounds of $R_T^{\sigma_\Lambda}$, but this time the termination condition relies on a new depth defined as follows.

\begin{definition}
Let $T$ be a rooted tree and $M>1$ be a constant. For any vertex $v$ in $T$, define the \concept{$M$-based depth} of $v$, denoted $\ell_M(v)$, as such: $\ell_M(v)=0$ if $v$ is the root, and $\ell_M(v)=\ell_M(u)+\lceil\log_M(d+1)\rceil$ if $v$ is one of the $d$ children of $u$.
\end{definition}

Let $M>1$ to be fixed. Denote by $B(\ell)$ the set of all vertices with $M$-based depth $<\ell$ along with their children and grandchildren in $T$.  It can be verified by induction that $|B(\ell)|\le n^2M^\ell$. The recursion is applied to estimate the $R_T^{\sigma_\Lambda}$ when the current $v\in B(\ell)$ until $v$ is no longer in $B(\ell)$ in which case the trivial bounds $0\le R_T^{\sigma_\Lambda}\le \infty$ is used.

Let $\epsilon_v$ be defined as in Section \ref{section-correlation-decay}.
Repeatedly applying Item \ref{lemma-epsilon-contract} of Lemma \ref{lemma-epsilon}, without loss of generality, we have a path $u_1u_2\cdots u_{k}$ in $T$ from the root $u_1=v$ to a $u_k$ with $\ell_M(u_k)\ge\ell$ and $\ell_M(u_{k-1})<\ell$, such that $\epsilon_{u_j}\le\alpha(d_j;x_1,\ldots,x_{d_j})\cdot\epsilon_{u_{j+1}}$ for $j=1,2,\ldots,k$, where $d_j$ is the number of children of $u_j$ and $x_i\in[0,\infty)$, $1\le i\le d_j$. The key to overcome the explosion caused by unbounded degrees is to observe that the contraction $\alpha(d;x_1,\ldots,x_{d})$ decreases dramatically as the degree $d$ grows.

\begin{lemma}\label{lemma-alpha-unbounded}
Let  $(\beta,\gamma,\lambda)$ be anti-ferromagnetic.
If $(\beta,\gamma,\lambda)$ satisfies the universal uniqueness, then there exist constants $\alpha<1$ and $M>1$ such that for any integer $d\ge 1$,  and any $x_i\in[0,\infty)$, $1\le i\le d$, it holds that $\alpha(d;x_1,...,x_d)\le\alpha^{\lceil\log_M (d+1)\rceil}$.
\end{lemma}
\begin{proof}
Assume the universal uniqueness of $(\beta,\gamma,\lambda)$. Due to Lemma \ref{lemma-asymmetric-alpha-bound}, there exists a constant $\alpha<1$ such that $\alpha(d;x_1,\ldots,x_{d})\le\alpha$. By Item \ref{lemma-threshold-infty} of Lemma \ref{lemma-threshold}, the universal uniqueness implies that $\gamma>1$, thus there exists a constant $M>1$ such that $d\cdot\sqrt{\frac{\lambda}{\gamma^{d+1}}}\le\alpha^{\lceil\log_M (d+1)\rceil}$ for all $d\ge M$. Due to \eqref{eq:alpha-unbounded}, it holds that  $\alpha(d;x_1,\ldots,x_{d})\le d\cdot \sqrt{\frac{\lambda}{\gamma^{d+1}}}\le\alpha^{\lceil\log_M (d+1)\rceil}$ for all $d\ge M$. Note that $\alpha(d;x_1,\ldots,x_{d})\le\alpha$ means that $\alpha(d;x_1,...,x_d)\le\alpha^{\lceil\log_M (d+1)\rceil}$ for $d<M$. Therefore, $\alpha(d;x_1,...,x_d)\le\alpha^{\lceil\log_M (d+1)\rceil}$ for all $d$.
\end{proof}

By Lemma \ref{lemma-alpha-unbounded}, there exist constants $\alpha<1$ and $M>1$ such that
\begin{align*}
\epsilon_v
&\le
\epsilon_{u_k}\cdot\prod_{j=1}^k\alpha^{\lceil\log_M (d_j+1)\rceil}\le \epsilon_{u_k}\cdot\alpha^{\sum_{j=1}^k\lceil\log_M (d_j+1)\rceil}=\epsilon_{u_k}\cdot\alpha^{\ell_M(u_k)}\le\epsilon_{u_k}\cdot\alpha^\ell.
\end{align*}
With the notation used in Section \ref{section-correlation-decay}, $S$ is the complement of $B(\ell)$.
Note that all $u_{k}$'s children are in $B(\ell)$ thus none of them are in $S$, and by Item \ref{lemma-threshold-infty} of Lemma \ref{lemma-threshold} the universal uniqueness implies that $\gamma>1$. Thus by Item \ref{lemma-epsilon-bound} of Lemma \ref{lemma-epsilon} it holds that $\epsilon_{u_{k}}=O(1)$. Therefore, $\epsilon_v\le\epsilon_{u_k}\cdot\alpha^\ell=O(\alpha^\ell)$.

Let $\delta_v=R_1-R_0$, where $R_0$ and $R_1$ are the bounds returned by the recursive procedure such that  $R_0\le R_T^{\sigma_\Lambda}\le R_1$.
By the same analysis as in Section \ref{section-correlation-decay}, $\delta_v=\frac{\epsilon_v}{\Phi(\widetilde{R})}=O(\epsilon_v)=O(\alpha^\ell)$.
Then by~\eqref{eq:p-bound}, the marginal probability $p_v^{\sigma_\Lambda}$ is approximated within an additive error of $O(\alpha^\ell)$. The running time of the recursion is $O(nB(\ell))=O(n^3M^\ell)$. By setting $\ell=\log_{\alpha}\epsilon$, we have an algorithm which approximates $p_v^{\sigma_\Lambda}$ within an additive error of $O(\epsilon)$ in time polynomial in $n$ and $\frac{1}{\epsilon}$, which implies an FPTAS for $Z(G)$ for arbitrary graph $G$.


\paragraph{Heterogeneous spin systems.} Our analysis in last and this sections actually holds for heterogeneous spin systems which allow that each vertex $v$ has a distinct constant external field $\lambda_v>0$. 

\begin{theorem}
For a two-state anti-ferromagnetic heterogeneous spin system with parameters $\beta$, $\gamma$, and external field $\lambda_v$ at each vertex $v$, for any finite $\Delta\ge 2$ or $\Delta=\infty$, if for all $v$ the $(\beta,\gamma,\lambda_v)$ is \unique{$\Delta$} then the spin system is of strong spatial mixing and has FPTAS on graphs of maximum degree at most $\Delta$.
\end{theorem}

\section{Non-monotonicity of general two-state spin system}\label{section-non-monotone}
In this section, we prove the following theorem.
\begin{theorem}
There exist two-state anti-ferromagnetic spin systems which exhibit strong spatial mixing on infinite $d$-regular tree but does not exhibit strong spatial mixing on all graphs of maximum degree at most $d$.
\end{theorem}
In a seminal work~\cite{Weitz06}, Weitz proved that for the hardcore model the strong spatial mixing on an infinite $d$-regular tree implies the strong spatial mixing on graphs of maximum degree at most $d$ (Theorem 2.3 in~\cite{Weitz06}). He further remarked that the same implication holds for all two-state spin systems.  That is,
\begin{quote}
for any two-state spin system, strong spatial mixing on an infinite $d$-regular tree implies the strong spatial mixing on graphs of maximum degree at most $d$. 
\end{quote}
This claim played important roles in current understanding of correlation decay in two-state spin systems as well as devising FPTAS for such systems.
An algorithmic form of this claim was cited in~\cite{SST} as a theorem for all two-state spin systems (Theorem 2.8 in~\cite{SST}) and was proved for the anti-ferromagnetic Ising model. It was raised as a conjecture in~\cite{sly2008uniqueness} whether this claim holds for multi-state spin systems.

Here we clarify that this claim holds only for the two-state spin systems under limited settings but does not hold for all general two-state spin systems.  This disproves the conjecture in~\cite{sly2008uniqueness} and shows that the common belief that the $d$-regular tree represents the worst case for strong spatial mixing among all graphs of maximum degree at most $d$ cannot be generalized to general two-state spin systems.
We first describe a region that the claim is true.

\begin{lemma}
For $0\le\beta,\gamma\le1$, the strong spatial mixing on infinite $d$-regular tree implies the strong spatial mixing on trees of maximum degree at most $d$.
\end{lemma}
\begin{proof}
Given a rooted tree of maximum degree at most $d$, for each vertex of $k$ children with $k<d-1$, we can attach $d-1-k$ dummy children with fixed (distributions of) spin states. This is the method used in~\cite{Weitz06} and~\cite{SST}: for the hardcore model the dummy children are fixed to be unoccupied and for the anti-ferromagnetic Ising model the dummy children are of uniform distributions over spin states. In both cases, the dummy children have no effect on their parent. In general, we fix the distribution to be $(p_0,p_1)$ at each dummy child satisfying 
\begin{align*}
p_0+p_1 &=1,\\
\beta p_0+p_1 &=p_0+\gamma p_1.
\end{align*}
When $0\le \beta,\gamma\le 1$, this system has solutions in $0\le p_0,p_1\le 1$.
With the ratio $R_i$ at the $i$-th child, $1\le i\le k$, and $R_i=p_0/p_1$ for the dummy children $k<i\le d-1$, the ratio at the parent is given by the recursion
\begin{align*}
\lambda\prod_{i=1}^{d-1}\frac{\beta R_i+1}{R_i+\gamma}
=
\lambda\prod_{i=1}^{k}\frac{\beta R_i+1}{R_i+\gamma},
\end{align*} 
which is identical to the original quantity.
\end{proof}

Due to the self-avoiding walk tree construction (Theorem \ref{theorem-T-saw}), it holds that for $0\le\beta,\gamma\le1$,  strong spatial mixing on infinite $d$-regular tree implies the strong spatial mixing and the FPTAS on graphs of maximum degree at most $d$.

For $0\le\beta,\gamma\le 1$, the spin system shows the following monotone property:
the uniqueness on infinite $d$-regular tree implies the uniqueness on all infinite regular trees of smaller degree. This can be verified by the following reasoning: due to Theorem \ref{theorem-regular-decay}, on infinite $d$-regular tree the uniqueness implies the strong spatial mixing, which for $0\le\beta,\gamma\le1$, implies the strong spatial mixing (including the uniqueness) on all infinite regular trees of smaller degree. 

There exist two-state anti-ferromagnetic spin systems which are non-monotone. 
We can choose anti-ferromagnetic $(\beta,\gamma,\lambda)$ satisfying that $\gamma>1$ and $\lambda>\lambda_c^{\mathsf{upper}}(\beta, \gamma)$, where $\lambda_c^{\mathsf{upper}}(\beta, \gamma)$ is the upper threshold for universal uniqueness given in Item~\ref{lemma-threshold-infty-exist} of Lemma~\ref{lemma-threshold}. Due to Item~\ref{lemma-threshold-non-monotone} of Lemma~\ref{lemma-threshold}, for $\gamma>1$ the uniqueness holds on $d$-regular trees for all sufficiently large $d$. On the other hand, due to Item~\ref{lemma-threshold-infty-exist} of Lemma~\ref{lemma-threshold}, $(\beta,\gamma,\lambda)$ cannot be universally unique when $\lambda>\lambda_c^{\mathsf{upper}}(\beta, \gamma)$, which means that there exists a finite $d'$ such that the system is non-unique on $d'$-regular tree.

For such non-monotone systems, due to Theorem \ref{theorem-regular-decay}, the uniqueness implies the strong spatial mixing on $d$-regular tree for sufficiently large $d$, but the strong spatial mixing does not hold on a regular tree of smaller degree (because of the non-uniqueness on the smaller tree). Therefore the implication between the strong spatial mixing on $d$-regular tree and on graphs of maximum degree at most $d$ does not hold for general two-state spin systems.



\paragraph{Acknowledgment.}  We would like to thank Alistair Sinclair and  Piyush Srivastava for the helpful comments to an early version of this paper.

\paragraph{Acknowledgment for the 2021 revision.}  We would like to thank Xiaoyu Chen, Weiming Feng and Xinyuan Zhang for locating an error in Lemma~\ref{lemma-threshold} in the original (2013) version of the paper and suggesting its fix.

\ifarxiv{

}
{
\bibliographystyle{abbrv}
\bibliography{paper}
}

\appendix

\section{The Uniqueness Thresholds}\label{app-threshold}
The following lemma translates the uniqueness condition into its various threshold forms. The lemma (except for item \ref{lemma-threshold-infty}) is not used in the proofs of the main results but is used in the interpretation of the main results and comparisons to the previous results which are mostly stated in threshold forms. 

\begin{lemma}\label{lemma-threshold}
Let  $(\beta,\gamma,\lambda)$ be anti-ferromagnetic.
\begin{enumerate}
\item\label{lemma-threshold-2-unique}
$(\beta,\gamma,\lambda)$ is \unique{2}.

\item \label{lemma-threshold-infty}
If  $\gamma\le 1$, then the uniqueness does not hold on infinite $d$-regular tree for all sufficiently large~$d$.

\item \label{lemma-threshold-non-monotone}
If $\gamma>1$, then the uniqueness holds on infinite $d$-regular tree for all sufficiently large~$d$.

\item \label{lemma-threshold-gamma}
For any $\Delta$ (including $\Delta=\infty$), there exists a critical threshold $\gamma_c=\gamma_c(\beta,\lambda,\Delta)$ such that  $(\beta, \gamma, \lambda)$ is \unique{$\Delta$} if and only if $\gamma\in (\gamma_c,\frac{1}{\beta})$. In particular, $\gamma_c(\beta,\lambda,\infty)>1$ and $\gamma_c(\beta,\lambda,\infty)=\gamma_c(\beta,\lambda,\Delta)$ for some finite $\Delta$.

\item \label{lemma-threshold-hardcore}
If $\beta = 0$, for any $\Delta$ (including $\Delta=\infty$), there exists a critical threshold $\lambda_c=\lambda_c(\gamma,\Delta)=\min_{1<d<\Delta} \frac{\gamma^{d+1}d^d}{(d-1)^{d+1}}$ such that $(\beta, \gamma, \lambda)$ is \unique{$\Delta$} if and only if $\lambda\in(0,\lambda_c)$. 

\item \label{lemma-threshold-free-external}
If  $\sqrt{\beta \gamma} > \frac{\Delta-2}{\Delta}$, then for any external field $\lambda$, $(\beta, \gamma, \lambda)$ is \unique{$\Delta$}.

\item  \label{lemma-threshold-external}
If $\beta > 0$, for any $\Delta$ (including $\Delta=\infty$) that $\sqrt{\beta \gamma} \le \frac{\Delta-2}{\Delta}$, the regime for the $\lambda$ satisfying up-to-$\Delta$ uniqueness is  as follows.

Let $\overline{\Delta}\triangleq\frac{1+\sqrt{\beta\gamma}}{1-\sqrt{\beta\gamma}}$.
For any integer $\overline{\Delta}\le d<\Delta$, let $x_1(d)\le x_2(d)$ be the two  positive roots of equation $\frac{d(1 - \beta\gamma)x}{(\beta x + 1)(x + \gamma)} = 1$, more specifically,
\begin{eqnarray*}
    x_1(d)= \frac{-1-\beta \gamma +d (1-\beta \gamma) - \sqrt{(-1-\beta \gamma +d (1-\beta \gamma))^2-4 \beta \gamma} }{2 \beta}, \\
    x_2(d)= \frac{-1-\beta \gamma +d (1-\beta \gamma) + \sqrt{(-1-\beta \gamma +d (1-\beta \gamma))^2-4 \beta \gamma} }{2 \beta}.
\end{eqnarray*}
Let $\lambda_i(d)\triangleq x_i(d) \left(\frac{x_i(d) +\gamma}{\beta x_i(d) +1}\right)^d$, where $i = 1, 2$, be defined for integers $\overline{\Delta}\le d<\Delta$.

Then $(\beta, \gamma, \lambda)$ is \unique{$\Delta$} if and only if $\lambda$ belongs to the following regime
\begin{align}
\bigcap_{\overline{\Delta}\le d<\Delta}\left((0,\lambda_1(d))\cup(\lambda_2(d),\infty)\right).\label{eq:uniqueness-regime-external-field}
\end{align}

And if furthermore $\gamma\le 1$, then $(\beta, \gamma, \lambda)$ is \unique{$\Delta$} if and only if $\lambda \in(0,\lambda_c)\cup(\bar{\lambda}_c,\infty)$, where 
\begin{align*}
\lambda_c
&=\lambda_c(\beta, \gamma, \Delta) \triangleq \min_{\overline{\Delta} \leqslant d < \Delta}\lambda_1(d),\\
\bar{\lambda}_c
&= \bar{\lambda}_c(\beta,\gamma,\Delta) \triangleq\max_{\overline{\Delta} \leqslant d < \Delta}\lambda_2(d)=\lambda_2(\Delta-1).
\end{align*}
In particular, when $\beta=\gamma$, it holds that $\lambda_c\cdot\bar{\lambda}_c=1$, and thus $(\beta, \beta, \lambda)$ is \unique{$\Delta$} if and only if $|\log \lambda|>\log\bar{\lambda}_c$.

\item \label{lemma-threshold-infty-exist}
$(\beta,\gamma,\lambda)$ can be universally unique only when $\gamma>1$. If $\gamma >1$,
there exist finite positive constants $\lambda_c^{\mathsf{lower}}=\lambda_c^{\mathsf{lower}}(\beta, \gamma)$ and $\lambda_c^{\mathsf{upper}}=\lambda_c^{\mathsf{upper}}(\beta, \gamma)$ where $\lambda_c^{\mathsf{lower}}\le \lambda_c^{\mathsf{upper}}$, such that $(\beta,\gamma,\lambda)$ is universally unique if $\lambda<\lambda_c^{\mathsf{lower}}$ and $(\beta,\gamma,\lambda)$ is universally unique only if $\lambda<\lambda_c^{\mathsf{upper}}$.

\end{enumerate}
\end{lemma}

\begin{remark*}[\textbf{An error in the SODA'13 version}]
In the earlier version of the paper, there was an error in Item~\ref{lemma-threshold-external} of Lemma~\ref{lemma-threshold} (which was Lemma~3.1 in the SODA'13 conference version~\cite{LLY13}).
In that version, the uniqueness regime~\eqref{eq:uniqueness-regime-external-field} was unconditionally and wrongly simplified to the bi-interval simple form $(0,\lambda_c)\cup(\bar{\lambda}_c,\infty)$, which actually holds more restrictively when both $\beta,\gamma\le 1$. 
However, when $\gamma>1$, there exist such $\beta\gamma<1$ that the uniqueness regime~\eqref{eq:uniqueness-regime-external-field} may become more complicated than consisting of at most two intervals.
For example, when $\beta=\frac{1}{122}$, $\gamma=100$, and $\Delta=23$, the regime for the $\lambda$ such that $(\beta,\gamma,\lambda)$ is up-to-$\Delta$ unique, consists of three intervals $(0,102.664)\cup(104.339,106.967)\cup(109.444,\infty)$.

The current version corrects the error and changes the statements of Item~\ref{lemma-threshold-external} and Item~\ref{lemma-threshold-infty-exist} of Lemma~\ref{lemma-threshold}.
We further remark that the error does not affect the main result of the paper, namely the implication from the up-to-$\Delta$ uniqueness to strong spatial mixing and FPTAS, because Lemma~\ref{lemma-threshold} is not used in the proofs of the main results (except for Item~\ref{lemma-threshold-infty}), but is only used for the interpretation of the uniqueness regime in various threshold forms and  comparison with other known results.
\end{remark*}

\begin{proof}[Proof of Lemma~\ref{lemma-threshold}]
Let $f_d(x)=\lambda\left(\frac{\beta x_d+1}{x_d+\gamma}\right)^d$ and $\hat{x}_d=f_d(\hat{x}_d)$ be the positive fixed point of $f_d(x)$.
Then
\begin{align*}
|f_d(\hat{x}_d)|=\frac{d(1-\beta\gamma)\hat{x}_d}{(\beta\hat{x}_d+1)(\hat{x}_d+\gamma)}.
\end{align*}
Let $(\beta,\gamma,\lambda)$ be anti-ferromagnetic. That is, $0\le \beta\le\gamma $, $\gamma>0$, and $\beta\gamma < 1$, thus $\beta < 1$. 
\begin{enumerate}
\item
It is easy to verify that for $\beta\gamma < 1$, $(\beta x + 1)(x + \gamma) - (1 - \beta\gamma) x>0$ for any positive $x>0$. Therefore, when $d=1$, we have $|f_1(\hat{x}_1)|=\frac{(1-\beta\gamma)\hat{x}_1}{(\beta\hat{x}_1+1)(\hat{x}_1+\gamma)}<1$, which means that $(\beta,\gamma,\lambda)$ is \unique{2}.

\item
For all sufficiently large $d$, it holds that
    \begin{align*}
        \lambda\beta^d\exp\left(\frac{d}{(1-\beta\gamma)d-3}\right)
        &<\frac{d(1-\beta\gamma)-3}{\beta}, \text{ and}\\
        \lambda\exp\left(-\frac{\gamma d}{d(1-\beta\gamma)-3}\right)
        &>\frac{\gamma}{d(1-\beta\gamma)-3}.
    \end{align*}
By contradiction, suppose that  $|f_d(\hat{x}_d)|=\frac{d(1-\beta\gamma)\hat{x}_d}{(\beta\hat{x}_d+1)(\hat{x}_d+\gamma)}\le1$. Then,
    \begin{align*}
        1\ge \frac{d(1-\beta\gamma)\hat{x}_d}{(\beta \hat{x}_d+1)(\hat{x}_d+\gamma)} = \frac{d(1-\beta\gamma)}{\beta \hat{x}_d+\frac{\gamma}{\hat{x}_d}+(1+\beta\gamma)}\ge  \frac{d(1-\beta\gamma)}{\beta \hat{x}_d+\frac{\gamma}{\hat{x}_d}+2}.
    \end{align*}

    \begin{list}{}{}
    \item[Case.1:] $\hat{x}_d\ge\gamma$. Then $\frac{\gamma}{\hat{x}_d}\le 1$. Thus,
    \begin{align*}
        1\ge  \frac{d(1-\beta\gamma)}{\beta \hat{x}_d+\frac{\gamma}{\hat{x}_d}+2} \ge \frac{d(1-\beta\gamma)}{\beta \hat{x}_d+3},
    \end{align*}
    which implies that $\hat{x}_d\ge \frac{d(1-\beta\gamma)-3}{\beta}$. However, it holds that
    \begin{align*}
        \hat{x}_d &=\lambda\left(\frac{\beta \hat{x}_d+1}{\hat{x}_d+\gamma}\right)^d
        \le\lambda\left(\frac{\beta \hat{x}_d+1}{\hat{x}_d}\right)^d
        \le\lambda\left(\beta+\frac{\beta}{d(1-\beta\gamma)-3}\right)^d\\
        &\le\lambda\beta^d\exp\left(\frac{d}{(1-\beta\gamma)d-3}\right)
        <\frac{d(1-\beta\gamma)-3}{\beta},
    \end{align*}
    a contradiction.

    \item[Case.2:] $\hat{x}_d<\gamma$. Then $\beta \hat{x}_d\le\beta\gamma<1$. Thus,
    \begin{align*}
        1\ge  \frac{d(1-\beta\gamma)}{\beta \hat{x}_d+\frac{\gamma}{\hat{x}_d}+2} \ge \frac{d(1-\beta\gamma)}{\frac{\gamma}{\hat{x}_d}+3},
    \end{align*}
    which implies that $\hat{x}_d\le \frac{\gamma}{d(1-\beta\gamma)-3}$. However, it holds that
    \begin{align*}
        \hat{x}_d &=\lambda\left(\frac{\beta \hat{x}_d+1}{\hat{x}_d+\gamma}\right)^d
        \ge\frac{\lambda}{(\hat{x}_d+1)^d}
        \ge\lambda\left(1+\frac{\gamma}{d(1-\beta\gamma)-3}\right)^{-d}\\
        &\ge\lambda\cdot\exp\left(-\frac{\gamma d}{d(1-\beta\gamma)-3}\right)
        >\frac{\gamma}{d(1-\beta\gamma)-3},
    \end{align*}
    a contradiction.
    \end{list}

\item
Let $\gamma>1$.
The fixed point $\hat{x}_d=\lambda\left(\frac{\beta\hat{x}_d+1}{\hat{x}_d+\gamma}\right)^d\le\frac{\lambda}{\gamma^d}$, and
\begin{align*}
|f_d'(\hat{x}_d)|=\frac{d(1-\beta\gamma)\hat{x}_d}{(\beta\hat{x}_d+1)(\hat{x}_d+\gamma)}\le\frac{d\lambda}{\gamma^d},
\end{align*}
which is strictly less than 1 for all sufficiently large $d$. Thus the uniqueness holds on infinite $d$-regular tree for all sufficiently large $d$.

\item
We first show that there exists a critical threshold $\gamma_c=\gamma_c(\beta,\lambda,\Delta)$ such that  $(\beta, \gamma, \lambda)$ is \unique{$\Delta$} if and only if $\gamma\in (\gamma_c,\frac{1}{\beta})$.
It is sufficient to show that if an anti-ferromagnetic $(\beta,\gamma,\lambda)$ is \unique{$\Delta$} then $(\beta,\gamma',\lambda)$ is \unique{$\Delta$} for any $\gamma' > \gamma$ and $\beta\gamma' < 1$. 

Recall that $\hat{x}_d$ is the positive fixed point of $f_d(x)=\lambda\left(\frac{\beta x+1}{x+\gamma}\right)^d$. 
Also let $\hat{x}_d'$ denote the positive solution to $x=\lambda\left(\frac{\beta x+1}{x+\gamma'}\right)^d$. 

We first show that $\hat{x}'_d< \hat{x}_d$. By contradiction, assume that $\hat{x}'_d\geq \hat{x}_d$. Since for anti-ferromagnetic $(\beta,\gamma,\lambda)$, $f_d(x)$ is monotonically decreasing, we have
\[ \hat{x}_d= \lambda\left(\frac{\beta \hat{x}_d +1}{\hat{x}_d+\gamma}\right)^d
\ge \lambda\left(\frac{\beta \hat{x}_d' +1}{\hat{x}_d'+\gamma}\right)^d
> \lambda\left(\frac{\beta \hat{x}_d' +1}{\hat{x}_d'+\gamma'}\right)^d
=\hat{x}'_d,\]
a contradiction.

Since
\[\lambda\left(\beta + \frac{(1-\beta\gamma')}{\hat{x}'_d+\gamma'}\right)^d=\hat{x}'_d<\hat{x}_d=\lambda\left(\beta + \frac{(1-\beta\gamma)}{\hat{x}_d+\gamma}\right)^d ,\]
we have 
\[\frac{(1-\beta\gamma')}{\hat{x}'_d+\gamma'} < \frac{(1-\beta\gamma)}{\hat{x}_d+\gamma}.\]
For $\hat{x}'_d<\hat{x}_d$, it also holds that
\[\frac{ \hat{x}'_d }{\beta \hat{x}'_d +1 }= \frac{1  }{\beta  +\frac{1}{\hat{x}'_d } }
<\frac{1}{\beta  +\frac{1}{\hat{x}_d } }=\frac{ \hat{x}_d }{\beta \hat{x}_d +1 }.\]
Multiplying the above two inequalities together, we have
\[\frac{ d(1-\beta\gamma')\hat{x}'_d }{ (\beta\hat{x}'_d+1) (\hat{x}_d+\gamma')}< \frac{ d(1-\beta\gamma)\hat{x}_d }{ (\beta\hat{x}_d+1) (\hat{x}_d+\gamma)}.\]
Note that these are the absolute derivatives at the respective fixed points when the parameters are $(\beta,\gamma,\lambda)$ and $(\beta,\gamma',\lambda)$.
Therefore if $(\beta,\gamma,\lambda)$ is \unique{$\Delta$} then $(\beta,\gamma',\lambda)$ is \unique{$\Delta$}.

Due to Part \ref{lemma-threshold-infty} of the lemma, if $\gamma\le 1$, $|f'_d(\hat{x}_d)|>1$ for all sufficiently large $d$, thus $\gamma_c(\beta,\lambda,\infty)>1$. And due to Part \ref{lemma-threshold-non-monotone} of the lemma, for any $\gamma\ge\gamma_c(\beta,\lambda,\infty)>1$, $|f'_d(\hat{x}_d)|$ is arbitrarily close to 0 as $d$ grows to infinity, thus $\gamma_c(\beta,\lambda,\infty)=\gamma_c(\beta,\lambda,\Delta)$ for a finite $\Delta$.

\item




When $\beta=0$, 
$|f_d'(\hat{x}_d)|=\frac{d(1-\beta\gamma)\hat{x}_d}{(\beta\hat{x}_d+1)(\hat{x}_d+\gamma)}=\frac{d\hat{x}_d}{\hat{x}_d + \gamma}$,
the uniqueness condition $|f'_d(\hat{x}_d)|<1$ is equivalent to that $\hat{x}_d< \frac{\gamma}{d-1}$ (here we assume $d>1$ since due to Part \ref{lemma-threshold-2-unique} of the lemma, for $d=1$ the uniqueness always holds). Recall that $ \hat{x}_d=\lambda\left(\frac{1}{\hat{x}_d + \gamma}\right)^d$. Then $|f_d'(\hat{x}_d)|<1$ if and only if
\begin{align*}
\lambda=\hat{x}_d (\hat{x}_d+\gamma)^d < \frac{\gamma^{d+1} d^d}{(d-1)^{d+1}}.
\end{align*}
Let $\lambda_c=\lambda_c(\gamma,\Delta)=\min_{1<d<\Delta} \frac{\gamma^{d+1}d^d}{(d-1)^{d+1}}$. It holds that $(\beta, \gamma, \lambda)$ is \unique{$\Delta$} if and only if $\lambda\in(0,\lambda_c)$.

\item
We note that $|f_d'(\hat{x}_d)|=\frac{d(1-\beta\gamma)\hat{x}_d}{(\beta\hat{x}_d+1)(\hat{x}_d+\gamma)}$ is not monotone in $\hat{x}_d$. It achieves its maximum value at $\hat{x}_d=\sqrt{\frac{\beta}{\gamma}}$.
Therefore, if for any $1\le d< \Delta$, $\frac{d(1-\beta\gamma)x}{(\beta x+1)(x+\gamma)}< 1$ for $x=\sqrt{\frac{\beta}{\gamma}}$, then $(\beta,\gamma,\lambda)$ is \unique{$\Delta$} for any $\lambda$. This condition holds when $\sqrt{\beta \gamma} > \frac{d - 1}{d + 1}$ for all $1\le d< \Delta$, i.e.~when $\sqrt{\beta \gamma} >\frac{\Delta-2}{\Delta}$.

\item \label{lemma-item-beta-larger-than-0}
Let $\theta(d)\triangleq-1-\beta \gamma +d (1-\beta \gamma)$.
It holds that $\theta(d)\ge2\sqrt{\beta \gamma}$ for all $d\ge\overline{\Delta}=\frac{1+\sqrt{\beta\gamma}}{1-\sqrt{\beta\gamma}}$.
Thus, $x_1(d),x_2(d)$ are well-defined and can be expressed as:
  \begin{eqnarray*}
    x_1(d)= \frac{\theta(d) - \sqrt{\theta(d)^2-4 \beta \gamma} }{2 \beta} \quad\text{ and }\quad
    x_2(d)= \frac{\theta(d) + \sqrt{\theta(d)^2-4 \beta \gamma} }{2 \beta}.
  \end{eqnarray*}
Recall that $\lambda_i(d)\triangleq x_i(d) \left(\frac{x_i(d) +\gamma}{\beta x_i(d) +1}\right)^d$ for $i = 1, 2$. 

The following monotonicity and unboundedness of $\lambda_2(d)$ is easy to verify.
\begin{claim}\label{claim-monotonicity-lambda2}
$ \lambda_2(d)$ is monotonically increasing in $d$ and goes to infinity as $d$ grows.
\end{claim}
\begin{proof}

Observe that for  $0\le\beta\le\gamma$, $\beta\gamma<1$ and $d\ge\overline{\Delta}$, it holds that  $\beta \le\sqrt{\beta\gamma}\le \frac{d - 1}{d + 1}< 1$ and $\frac{x + \gamma}{\beta x + 1}$ is increasing in $x$.  Moreover,
\begin{eqnarray*}
\frac{x_2(d) + \gamma}{\beta x_2(d) + 1} &=& \frac{1}{\beta}\cdot\frac{(d-1)(1-\beta\gamma)+\sqrt{\theta(d)^2 - 4\beta\gamma}}{(d + 1)(1 - \beta\gamma) + \sqrt{\theta(d)^2 - 4\beta\gamma}}\\
&\geq& \frac{d + 1}{d - 1}\cdot\frac{(d - 1)(1 - \beta\gamma)}{(d + 1)(1 - \beta\gamma)} = 1.
\end{eqnarray*}

It is easy to verify that $x_2(d)$ is increasing in $d$ and unbounded as $d$ grows. 
Therefore, we can conclude that 
$\lambda_2(d) = x_2(d)\left(\frac{x_2(d) + \gamma}{\beta x_2(d) + 1}\right)^d$ is increasing in $d$ and is unbounded as $d$ grows.
\end{proof}

We proceed to analyze the uniqueness regime.
For $1\le d< \overline{\Delta} $, it holds that $\sqrt{\beta\gamma} > \frac {d - 1}{d + 1}$, in which case it always holds that $|f'_d(\hat{x}_d)|<1$ due to Part \ref{lemma-threshold-free-external} of this lemma.

Recall that $\sqrt{\beta\gamma} \le \frac{\Delta - 2}{\Delta}$. 
It remains to analyze the $|f'_d(\hat{x}_d)|$ for those $\overline{\Delta}\le d<\Delta$. 
For such $d$'s, we have $\sqrt{\beta\gamma}\le \frac{d - 1}{d + 1}$,
and since $x_1(d)$ and $x_2(d)$ are the roots of equation $\frac{d(1-\beta\gamma)x}{(\beta x+1)(x+\gamma)}=1$,
it holds  that $|f'_d(\hat{x}_d)|=\frac{d(1 - \beta\gamma)\hat{x}_d}{(\beta\hat{x}_d + 1)(\hat{x}_d + \gamma)}<1$ if and only if  $\hat{x}_d<x_1(d)$ or $\hat{x}_d>x_2(d)$. 
Note that $x\left(\frac{x + \gamma}{\beta x + 1}\right)^d$ is monotonically increasing in $x$. Thus:
\begin{align*}
\hat{x}_d<x_1(d) 
\quad&\Longleftrightarrow\quad \lambda=\hat{x}_d\left(\frac{\hat{x}_d + \gamma}{\beta\hat{x}_d + 1}\right)^d<x_1(d) \left(\frac{x_1(d) +\gamma}{\beta x_1(d) +1}\right)^d=\lambda_1(d),\\
\hat{x}_d>x_2(d)
\quad&\Longleftrightarrow\quad \lambda=\hat{x}_d\left(\frac{\hat{x}_d + \gamma}{\beta\hat{x}_d + 1}\right)^d>x_2(d) \left(\frac{x_2(d) +\gamma}{\beta x_2(d) +1}\right)^d=\lambda_2(d),
\end{align*}
which means for any $\overline{\Delta}\le d<\Delta$, $|f'_d(\hat{x}_d)|<1$ if and only if $\lambda\in(0,\lambda_1(d))\cup(\lambda_2(d),\infty)$.

Therefore, $(\beta, \gamma, \lambda)$ is \unique{$\Delta$} if and only if $\lambda$ belongs to the regime
\begin{align*}
\bigcap_{\overline{\Delta}\leqslant d<\Delta}\left((0,\lambda_1(d))\cup(\lambda_2(d),\infty)\right), 
\end{align*}
which is exactly the uniqueness regime~\eqref{eq:uniqueness-regime-external-field}.
Recall that
\begin{align*}
\lambda_c
=\lambda_c(\beta, \gamma, \Delta) \triangleq \min_{\overline{\Delta} \leqslant d < \Delta}\lambda_1(d)\quad\text{ and }\quad
\bar{\lambda}_c
= \bar{\lambda}_c(\beta,\gamma,\Delta) \triangleq \max_{\overline{\Delta} \leqslant d < \Delta}\lambda_2(d)=\lambda_2(\Delta-1),
\end{align*}
where the last equation is due to the monotonicity of $\lambda_2(d)$. 


We then show that for $\gamma\le 1$, the uniqueness regime~\eqref{eq:uniqueness-regime-external-field} becomes $(0,\lambda_c)\cup(\bar{\lambda}_c,\infty)$. 
We need  the following claims, whose proofs are postponed to the end of this section.
\begin{claim}\label{claim-uniqueness-regime-2-interval}
For $\overline{\Delta}\le d_0< d_1\le\infty$, it always holds that 
\[
\left(0, \min_{d_0\leqslant d< d_1}\lambda_1(d)\right)\cup\left( \max_{d_0\leqslant d< d_1}\lambda_2(d),\infty\right)
\subseteq \bigcap_{d_0\leqslant d< d_1}\left((0,\lambda_1(d))\cup(\lambda_2(d),\infty)\right);
\]
furthermore, if $\lambda_1(d+1)\le \lambda_2(d)$ for all integers $d_0\le d<d_1-1$, then
\[
\left(0, \min_{d_0\leqslant d< d_1}\lambda_1(d)\right)\cup\left( \max_{d_0\leqslant d< d_1}\lambda_2(d),\infty\right)=\bigcap_{d_0\leqslant d< d_1}\left((0,\lambda_1(d))\cup(\lambda_2(d),\infty)\right).
\]
\end{claim}

\begin{claim}\label{claim-uniqueness-regime-sandwich}
The followings hold:
\begin{itemize}
\item if $\gamma\le 1$, then $\lambda_1(d+1)\le \lambda_2(d)$ for all integers $d\ge\overline{\Delta}$;
\item if $\gamma>1$, there is a finite $d_0=d_0(\beta,\gamma)\ge\overline{\Delta}$ such that $\lambda_1(d+1)\le \lambda_2(d)$ for all $d\ge d_0$.
\end{itemize}
\end{claim}

Now assume that $\gamma\le 1$.
By Claim~\ref{claim-uniqueness-regime-sandwich},  $\lambda_1(d+1)\le \lambda_2(d)$ for all integers $\overline{\Delta}\le d<\Delta-1$.
Then due to Claim~\ref{claim-uniqueness-regime-2-interval}, the uniqueness regime~\eqref{eq:uniqueness-regime-external-field} becomes 
\[
\bigcap_{\overline{\Delta}\leqslant d< \Delta}\left((0,\lambda_1(d))\cup(\lambda_2(d),\infty)\right)=\left(0, \lambda_c\right)\cup\left( \bar{\lambda}_c,\infty\right).
\]

In particular, consider $\beta=\gamma$, the anti-ferromagnetic Ising model.
It is easy to verify that $x_1(d)\cdot x_2(d)=\frac{\gamma}{\beta}=1$, and thus 
\[
\lambda_1(d)\cdot\lambda_2(d)=x_1(d)\cdot x_2(d) \left(\frac{x_1(d) +\gamma}{\beta x_1(d) +1}\cdot\frac{x_2(d) +\gamma}{\beta x_2(d) +1}\right)^d=1.
\] 
Therefore, $\lambda_1(d)=1/\lambda_2(d)$ is monotonically decreasing, due to the monotonicity of $\lambda_2(d)$ guaranteed by Claim~\ref{claim-monotonicity-lambda2}.
As a result, $\lambda_c=\lambda_1(\Delta-1)$, $\bar{\lambda}_c=\lambda_2(\Delta-1)$, and $\lambda_c\cdot\bar{\lambda}_c=\lambda_1(\Delta-1)\cdot\lambda_2(\Delta-1)=1$.
Then $(\beta,\beta,\lambda)$ is \unique{$\Delta$} if and only if $|\log\lambda|>\log\bar{\lambda}_c$.

\item
The case of $\beta=0$ is taken cared of in Part~\ref{lemma-threshold-hardcore} of this lemma. In such case, due to Part~\ref{lemma-threshold-hardcore} of the lemma, $(0, \gamma, \lambda)$ is \unique{$\Delta$} if and only if $\lambda\in(0,\lambda_c(\gamma,\Delta))$ where $\lambda_c(\gamma,\Delta)=\min_{1<d<\Delta} \frac{\gamma^{d+1}d^d}{(d-1)^{d+1}}$.
In particular, $\lambda_c(\gamma,\infty)=\min_{d>1} \frac{\gamma^{d+1}d^d}{(d-1)^{d+1}}$ becomes 0 if $\gamma\le 1$ and becomes a positive and finite $\lambda_c(\gamma)>0$ if $\gamma>1$. Therefore, $(0, \gamma, \lambda)$ is universally unique only if $\gamma>1$, and when $\gamma>1$, $(0, \gamma, \lambda)$ is universally unique if and only if $\lambda\in(0,\lambda_c(\gamma))$. We can let $\lambda_c^{\mathsf{lower}}(\gamma)=\lambda_c^{\mathsf{upper}}(\gamma)=\lambda_c(\gamma)$ and the lemma follows.

For the case of $\beta>0$, due to Part~\ref{lemma-threshold-external} of the lemma, for all $\Delta\ge \overline{\Delta}\triangleq\frac{1+\sqrt{\beta\gamma}}{1-\sqrt{\beta\gamma}}$, the tuple $(\beta, \gamma, \lambda)$ is \unique{$\Delta$} if and only if  $\lambda$ belongs to the following regime:
\begin{align*}
\bigcap_{\overline{\Delta}\leqslant d<\Delta}\left((0,\lambda_1(d))\cup(\lambda_2(d),\infty)\right). 
\end{align*}
When $\gamma\le 1$, 
this regime becomes  $(0,\lambda_c(\beta,\gamma,\Delta))\cup (\bar{\lambda}(\beta,\gamma,\Delta),\infty)$, where $\lambda_c(\beta, \gamma, \Delta) = \min_{\overline{\Delta} \leqslant d < \Delta}\lambda_1(d)$ and $\bar{\lambda}_c(\beta,\gamma,\Delta) = \max_{\overline{\Delta} \leqslant d < \Delta}\lambda_2(d)=\lambda_2(\Delta-1)$. 

\begin{claim}\label{claim-uniqueness-regime-lambda1-limit}
The followings hold:
\begin{itemize}
\item if $\gamma\le 1$, then $\lambda_1(d)$ approaches 0 as $d$ grows to infinity;
\item if $\gamma> 1$, then $\lambda_1(d)$ is unbounded as $d$ grows to infinity.
\end{itemize}
\end{claim}
\begin{proof}
For integers $d\ge\overline{\Delta}$, let $\theta(d)\triangleq-1-\beta \gamma +d (1-\beta \gamma)$, and let 
\begin{align}
\hat{\lambda}_1(d)\triangleq \frac{2\gamma^{d+1}}{\theta(d)}\left(\frac{d+1}{d-1}\right)^d.\label{eq:uniqueness-regime-lambda-1-hat}
\end{align}
When $\gamma\le 1$, it is easy to verify that $\hat{\lambda}_1(d)$ approaches 0 as $d$ grows to infinity.

Moreover, $\lambda_1(d)$ is upper bounded by $\hat{\lambda}_1(d)$. Specifically, for $d\ge\overline{\Delta}$,
\begin{align*}
x_1(d)
&=
\frac{\theta(d) - \sqrt{\theta(d)^2-4 \beta \gamma} }{2 \beta}
=
\frac{2\gamma}{\theta(d) +\sqrt{\theta(d)^2-4 \beta \gamma} }
\le
\frac{2\gamma}{\theta(d)}.
\end{align*}
Since $\frac{x+\gamma}{\beta x+1}$ is increasing in $x$ when $\beta\gamma<1$, we have
\[
\frac{x_1(d)+\gamma}{\beta x_1(d)+1}
\le
\frac{\frac{2\gamma}{\theta(d)}+\gamma}{\frac{2\beta\gamma}{\theta(d)}+1}
=
\gamma\cdot\frac{d+1}{d-1}.
\]
Thus, for all  $d\ge\overline{\Delta}$,
\begin{align}
\lambda_1(d) 
&= x_1(d)\left(\frac{x_1(d) + \gamma}{\beta x_1(d) + 1}\right)^d
\le
\frac{2\gamma^{d+1}}{\theta(d)}\left(\frac{d+1}{d-1}\right)^d
=\hat{\lambda}_1(d).\label{eq:uniqueness-regime-lambda-1-upper-bound}
\end{align}
Therefore, when $\gamma\le 1$, $\lambda_1(d)$ approaches 0 as $d$ grows to infinity.

On the other hand, when $\gamma > 1$, for $d\ge\overline{\Delta}$,
  \begin{align*}
    \lambda_1(d) 
    = x_1(d)\left(\frac{x_1(d) + \gamma}{\beta x_1(d) + 1}\right)^d
    \geq x_1(d)\cdot\gamma^d
    = \frac{2\gamma^{d+1}}{\theta(d) + \sqrt{\theta(d)^2-4 \beta \gamma}},
  \end{align*}
which is obviously unbounded as $d$ grows.
\end{proof}

When $\lambda\le 1$, due to above claim, $\lambda_1(d)$ approaches 0 as $d$ grows to infinity, which means $\lambda_c(\beta,\gamma,\infty)=0$; and by Claim~\ref{claim-monotonicity-lambda2},  $\lambda_2(d)$ goes to infinity as $d$ grows to infinity, which means $\bar{\lambda}_c(\beta,\gamma,\infty)=\infty$. Therefore, $(\beta, \gamma, \lambda)$ can be universally unique only when $\gamma>1$. 

Now assume $\gamma>1$. By Claim~\ref{claim-uniqueness-regime-lambda1-limit}, $\lambda_1(d)$ is unbounded as $d$ grows.
Thus, $\lambda_1(d)$ achieves its minimum value at a finite $d=d(\beta,\gamma)\ge\overline{\Delta}$.
And it is obvious that $\lambda_1(d)$ is always finite and positive for any finite $d\ge \overline{\Delta}$.
Therefore, $\lambda_c(\beta, \gamma, \infty) = \lambda_c(\beta, \gamma, d(\beta,\gamma))$ is finite and positive.
Let $\lambda_c^{\mathsf{lower}}(\beta, \gamma)=\lambda_c(\beta, \gamma, \infty)>0$ denote this finite positive threshold.

On the other hand, by Claim~\ref{claim-uniqueness-regime-sandwich}, for $\gamma>1$, there is a finite $d_0=d_0(\beta,\gamma)\ge\overline{\Delta}$ such that $\lambda_1(d+1)\le\lambda_2(d)$ for all $d\ge d_0$.
Let $\lambda_c^{\mathsf{upper}}(\beta, \gamma)=\min_{d\ge d_0}\lambda_1(d)$.
Obviously, $\lambda_c^{\mathsf{upper}}(\beta, \gamma)\ge \lambda_c^{\mathsf{lower}}(\beta, \gamma)>0$. And $\lambda_c^{\mathsf{upper}}(\beta, \gamma)\le \lambda_1(d_0(\beta,\gamma))$ is also finite.

Due to Claim~\ref{claim-uniqueness-regime-2-interval} and the monotonicity and unboundedness of $\lambda_2(d)$ by Claim~\ref{claim-monotonicity-lambda2}, the regime for $\infty$-uniqueness, namely $\bigcap_{ d\geqslant\overline{\Delta}}\left((0,\lambda_1(d))\cup(\lambda_2(d),\infty)\right)$, is sandwiched as follows:
\begin{align*}
\left(0,\lambda_c^{\mathsf{lower}}(\beta, \gamma)\right)
&\subseteq
\bigcap_{ d\geqslant\overline{\Delta}}\left((0,\lambda_1(d))\cup(\lambda_2(d),\infty)\right)\\
&\subseteq 
\bigcap_{d\geqslant d_0}\left((0,\lambda_1(d))\cup(\lambda_2(d),\infty)\right)\\
&=
\left(0,\lambda_c^{\mathsf{upper}}(\beta, \gamma)\right)\cup(\lambda_2(\infty),\infty)\\
&=\left(0,\lambda_c^{\mathsf{upper}}(\beta, \gamma)\right).
\end{align*}
Therefore, 
$(\beta, \gamma, \lambda)$ is universally unique, if $\lambda<\lambda_c^{\mathsf{lower}}(\beta, \gamma)$
and only if $\lambda<\lambda_c^{\mathsf{upper}}(\beta, \gamma)$.
\end{enumerate}
\end{proof}

It now remains to prove Claim~\ref{claim-uniqueness-regime-2-interval} and Claim~\ref{claim-uniqueness-regime-sandwich}.
\begin{proof}[Proof of Claim~\ref{claim-uniqueness-regime-2-interval}]
First, it is easy to verify that for $\overline{\Delta}\le d_0< d_1\le\infty$, it always holds that
\[
\left(0, \min_{d_0\leqslant d< d_1}\lambda_1(d)\right)\cup\left( \max_{d_0\leqslant d< d_1}\lambda_2(d),\infty\right)
\subseteq \bigcap_{d_0\leqslant d< d_1}\left((0,\lambda_1(d))\cup(\lambda_2(d),\infty)\right),
\]
since for all $d_0\le d< d_1$, it always holds that 
\[
\left(0, \min_{d_0\leqslant d< d_1}\lambda_1(d)\right)\subseteq (0,\lambda_1(d))
\quad\text{ and }\quad\left( \max_{d_0\leqslant d< d_1}\lambda_2(d),\infty\right)\subseteq (\lambda_2(d),\infty).
\]

Next, we prove that the two regimes are equal if $\lambda_1(d+1)\le \lambda_2(d)$ for all integers $d_0\le d<d_1-1$.
This is proved by induction.

The equality in the claim is trivial to hold when $d_0< d_1\le d_0+1$, because in this case there is at most one integer $d$ satisfying $d_0\le d<d_1$. Then both sides of the equality becomes
\[
\left((0,\lambda_1(d))\cup(\lambda_2(d),\infty)\right).
\]
This establishes the induction basis.

For the induction hypothesis, assume that the equality holds for $d_1>d_0$:
\begin{align}
\bigcap_{d_0\leqslant d< d_1}\left((0,\lambda_1(d))\cup(\lambda_2(d),\infty)\right)=\left(0, \lambda_1^{< d_1}\right)\cup\left( \lambda_2^{< d_1},\infty\right),\label{eq:uniqueness-regime-2-interval}
\end{align}
where we denote 
\[
\lambda_1^{< d_1}\triangleq\min_{d_0\leqslant d< d_1}\lambda_1(d)\quad\text{ and }\quad\lambda_2^{< d_1}\triangleq\max_{d_0\leqslant d< d_1}\lambda_2(d).
\]

For the induction step, we want to establish the equality when $d_1$ is increased by 1.
Thus in addition, we assume that $\lambda_1(d+1)\le \lambda_2(d)$ for the unique integer $d\in[d_1-1,d_1)$, that is
\begin{align}
\lambda_1(\left\lceil d_1\right\rceil)\le \lambda_2(\left\lceil d_1\right\rceil-1). \label{eq:uniqueness-regime-2-interval-assumption}
\end{align}
To finish the induction, it is then sufficient to verify that 
\[
\bigcap_{d_0\leqslant d< d_1+1}\left((0,\lambda_1(d))\cup(\lambda_2(d),\infty)\right)=\left(0, \lambda_1^{< d_1+1}\right)\cup\left( \lambda_2^{< d_1+1},\infty\right),
\]
assuming the induction hypothesis~\eqref{eq:uniqueness-regime-2-interval}. Therefore, it is sufficient to verify that 
\begin{align}
&\left(\left(0, \lambda_1^{< d_1}\right)\cup\left( \lambda_2^{< d_1},\infty\right)\right)\cap
\left((0,\lambda_1(\left\lceil d_1\right\rceil))\cup(\lambda_2(\left\lceil d_1\right\rceil),\infty)\right)\notag\\
=&\left(0, \lambda_1^{< d_1+1}\right)\cup\left( \lambda_2^{< d_1+1},\infty\right).\label{eq:uniqueness-regime-2-interval-induction}
\end{align}
Note that $\left\lceil d_1\right\rceil$ is the unique integer in $[d_1,d_1+1)$.
And due to the monotonicity of $\lambda_2(d)$, we have $\lambda_2(\left\lceil d_1\right\rceil)\ge  \lambda_2^{< d_1}=\lambda_2(\left\lceil d_1\right\rceil-1)$.
Therefore, $(\lambda_2(\left\lceil d_1\right\rceil),\infty)\subseteq \left( \lambda_2^{< d_1},\infty\right)$, and thus~\eqref{eq:uniqueness-regime-2-interval-induction} holds as long as $\lambda_1(\left\lceil d_1\right\rceil)\le \lambda_2^{< d_1}=\lambda_2(\left\lceil d_1\right\rceil-1)$, which is guaranteed by the assumption~\eqref{eq:uniqueness-regime-2-interval-assumption}.
\end{proof}

\begin{proof}[Proof of Claim~\ref{claim-uniqueness-regime-sandwich}]
Recall function $\hat{\lambda}_1(d)$ defined in~\eqref{eq:uniqueness-regime-lambda-1-hat}: For integers $d\ge\overline{\Delta}$, let $\hat{\lambda}_1(d)\triangleq \frac{2\gamma^{d+1}}{\theta(d)}\left(\frac{d+1}{d-1}\right)^d$, where $\theta(d)\triangleq-1-\beta \gamma +d (1-\beta \gamma)$.
By~\eqref{eq:uniqueness-regime-lambda-1-upper-bound}, for all $d\ge\overline{\Delta}$, we have 
$\lambda_1(d)\le \hat{\lambda}_1(d)$.

It also holds that $\hat{\lambda}_1(d)\le\lambda_2(d)$.
For all $d\ge\overline{\Delta}$, we have
$x_2(d)=\frac{\theta(d) + \sqrt{\theta(d)^2-4 \beta \gamma} }{2 \beta}\ge\frac{\theta(d)}{2\beta}$.
Since $\frac{x+\gamma}{\beta x+1}$ is monotonically increasing in $x$ when $\beta\gamma<1$, it holds that
$\frac{x_1(d)+\gamma}{\beta x_1(d)+1}\ge
\frac{\frac{\theta(d)}{2\beta}+\gamma}{\frac{\theta(d)}{2}+1}=
\frac{1}{\beta}\cdot\frac{d-1}{d+1}$.
Meanwhile, for all  $d\ge\overline{\Delta}=\frac{1+\sqrt{\beta\gamma}}{1-\sqrt{\beta\gamma}}$, it holds that $\beta\gamma\le\left(\frac{d-1}{d+1}\right)^2$ and $\theta(d)^2\ge4\beta\gamma$.
Therefore, 
\begin{align}
\lambda_2(d) 
&= x_2(d)\left(\frac{x_2(d) + \gamma}{\beta x_2(d) + 1}\right)^d
\ge
\frac{\theta(d)}{2\beta^{d+1}}\left(\frac{d-1}{d+1}\right)^d
\ge
\frac{2\gamma^{d+1}}{\theta(d)}\left(\frac{d+1}{d-1}\right)^d
=
\hat{\lambda}_1(d).\label{eq:uniqueness-regime-hat-lambda-1-lambda-2}
\end{align}
Assume $\gamma\le 1$.
It is also not hard to observe that $\hat{\lambda}_1(d)$ is decreasing in $d\ge\overline{\Delta}>1$.
Indeed, if $\gamma\le 1$, then $\frac{2\gamma^{d+1}}{\theta(d)}>0$ is monotonically decreasing in $d\ge\overline{\Delta}$; and we can verify the monotonicity of function $h(d)\triangleq\left(\frac{d+1}{d-1}\right)^d$ for $d\ge\overline{\Delta}>1$ as follows:
\begin{align*}
\left(\ln h(d)\right)'
&= 
-\frac{2d}{d^2-1}+\ln\left(\frac{d+1}{d-1}\right) \to 0 \text{ as } d\to\infty,\\
\left(\ln h(d)\right)''
&=
\frac{4}{d^2-1}>0,
\end{align*}
which guarantees that $h(d)$ is monotonically decreasing in $d$, and hence $\hat{\lambda}_1(d)=\frac{2\gamma^{d+1}}{\theta(d)}\cdot h(d)$ is decreasing in $d$ for $d\ge\overline{\Delta}$, since both $\frac{2\gamma^{d+1}}{\theta(d)}$ and $h(d)$ are positive and decreasing in $d$.


Therefore, when $\gamma\le 1$, due to~\eqref{eq:uniqueness-regime-lambda-1-upper-bound}, the monotonicity of $\hat{\lambda}_1(d)$, and~\eqref{eq:uniqueness-regime-hat-lambda-1-lambda-2}, 
for all $d\ge\overline{\Delta}$,
\[
\lambda_1(d+1)\le \hat{\lambda}_1(d+1)\le \hat{\lambda}_1(d)\le \lambda_2(d).
\]
This proves that $\lambda_1(d+1)\le \lambda_2(d)$ for all $d\ge\overline{\Delta}$ when $\gamma\le 1$.

Finally, when $\gamma> 1$, by contradiction suppose that $\lambda_1(d+1)> \lambda_2(d)$ for some $d\ge\overline{\Delta}$. We then deduce a finite upper bound $d_0(\beta,\gamma)$ on such bad $d$. 
By~\eqref{eq:uniqueness-regime-lambda-1-upper-bound}, we have $\lambda_1(d+1)\le\hat{\lambda}_1(d+1)=\frac{2\gamma^{d+2}}{\theta(d+1)}\cdot h(d+1)\le \frac{2\gamma^{d+2}}{\theta(d)}\cdot h(d)$, where the last inequality is due to that $\theta(d)$ is increasing and $h(d)$ is decreasing; and by~\eqref{eq:uniqueness-regime-hat-lambda-1-lambda-2}, we have $\lambda_2(d)\ge\frac{\theta(d)}{2\beta^{d+1}}\cdot \frac{1}{h(d)}$.
Hence  $\lambda_1(d+1)> \lambda_2(d)$ implies that
\begin{align}
\left(\frac{1}{\beta\gamma}\right)^{d}\frac{\theta(d)^2}{h(d)^2}
<4\beta\gamma^2.\label{eq:uniqueness-regime-bad-d-upper-bound}
\end{align}
Recall that for $d\ge\overline{\Delta}=\frac{1+\sqrt{\beta\gamma}}{1-\sqrt{\beta\gamma}}$, it always holds that $\theta(d)^2\ge 4\beta\gamma$. And without loss of generality, we only consider large enough $d\ge 2$, and thus $h(d)\le h(2)=9$ since $h(d)$ is decreasing. Therefore inequality~\eqref{eq:uniqueness-regime-bad-d-upper-bound} implies that
\[
\left(\frac{1}{\beta\gamma}\right)^{d}<81\gamma,
\]
which gives us a necessary condition $d<\log_{{1}/{\beta\gamma}}(81\gamma)$  for $\lambda_{1}(d+1)>\lambda_2(d)$, among the $d$'s that $d\ge\overline{\Delta}$ and $d\ge 2$.
Let 
\[
d_0(\beta,\gamma)\triangleq\max\left\{2, \overline{\Delta}, \log_{\frac{1}{\beta\gamma}}(81\gamma)\right\}.
\]
It then must hold that $\lambda_{1}(d+1)\le \lambda_2(d)$ for all $d\ge d_0(\beta,\gamma)$.
\end{proof}

\section{Proof of  Lemma \ref{lemma-epsilon} (a Mean Value Theorem approach)}\label{app-MVT}
We prove Lemma \ref{lemma-epsilon}. The notations are defined in Section \ref{section-correlation-decay}.


\ProofLemmaEpsilon

\section{Finding a good potential function heuristically}\label{section-potential}
Perhaps the most mysterious step in our proof is the choice of the potential function $\Phi(x)=\frac{1}{\sqrt{x (\beta x +1)(x+\gamma)}}$. 
As in many cases where potential analysis is applied, there is no standard routine for searching for a suitable potential function. On the other hand, it is quite unlikely that we can just guess such a fairly complicated formula without any hints. Here we present a heuristic approach which leads us to the discovery of a good potential function. This part is not rigorous and logically unnecessary for the soundness of our result. Nevertheless, this heuristic approach is general and interesting enough and may find its applications in other scenarios, thus deserves an exposition here.

The heuristics consists of three steps:
\begin{enumerate}
  \item Find a necessary condition for the potential function, which is an equation related to the potential function at one point.
  \item (heuristic step) Enhance the condition by assuming that the equation holds for the whole range of the variable, which gives us a differential equation.
  \item Solve the differential equation and get a potential function.
\end{enumerate}


We first assume that the system is at the boundary of uniqueness. This means that $d$ is the  critical degree and
$f'(\hat{x})= -1$, where $f(x)=\lambda\left(\frac{\beta x+1}{x+\gamma}\right)^d$ and $\hat{x}=f(\hat{x})$ is the positive fixed point.
%

We have the following identities:
\begin{align*}
\hat{x}
=\lambda\left(\frac{\beta\hat{x}+1}{\hat{x}+\gamma}\right)^d
\text{ and }1=\lambda d(1-\beta\gamma)\frac{(\beta\hat{x}+1)^{d-1}}{(\hat{x}+\gamma)^{d+1}},
\end{align*}
which together implies another identity
\begin{align}
d(1-\beta\gamma)\hat{x}
&=(\beta\hat{x}+1)(\hat{x}+\gamma).\label{eq:heuristic-identity}
\end{align}

The goal is to find a potential function such that $\alpha(x)=|f'(x)|\frac{\Phi\left(f(x)\right)}{\Phi(x)}\leq 1$ for all $x$.
On the other hand, we have
$\alpha(\hat{x})=|f'(\hat{x})|\frac{\Phi(f(\hat{x}))}{\Phi(\hat{x})}=1\cdot \frac{\Phi(\hat{x})}{\Phi(\hat{x})}=1$.
So $\alpha(x)$ achieves its maximum when $x=\hat{x}$.
As a differentiable function, it holds that $\alpha'(\hat{x})=0$, that is
\[\left.\left[f'(x)\cdot\frac{\Phi\left(f(x)\right)}{\Phi(x)}\right]'\right|_{x=\hat{x}}=0.\]

We have the following calculation:
\begin{eqnarray*}
& & \left[f'(x)\cdot\frac{\Phi(f(x))}{\Phi(x)}\right]'_{x=\hat{x}}=0 \\
&\Leftrightarrow&\left.\left[f'(x)\Phi(f(x))\right]'\Phi(x)\right|_{x=\hat{x}}=\left.f'(x)\Phi(f(x))\Phi'(x)\right|_{x=\hat{x}}\nonumber\\
&\Leftrightarrow&\left[f''(\hat{x})\Phi(f(\hat{x}))+f'(\hat{x})\Phi'(f(\hat{x}))f'(\hat{x})\right]\Phi(\hat{x})=f'(\hat{x})\Phi(f(\hat{x}))\Phi'(\hat{x})\\
&\Leftrightarrow&f''(\hat{x})\Phi(\hat{x})+\Phi'(\hat{x})=-\Phi'(\hat{x})\\
&\Leftrightarrow&-\frac{f''(\hat{x})}{2}=\frac{\Phi'(\hat{x})}{\Phi(\hat{x})}=\left(\ln(\Phi(\hat{x}))\right)',
\end{eqnarray*}
where we use the fact that $\hat{x}=f(\hat{x})$ and $f'(\hat{x})=-1$.

Taking the second derivative of $f(x)$, we have
\begin{eqnarray*}
f''(x)=\lambda d(\beta\gamma-1)\frac{(\beta x+1)^{d-2}}{(x+\gamma)^{d+2}}\left((d-1)\beta(x+\gamma)-(d+1)(\beta x+1)\right).
\end{eqnarray*}
This equation already gives an identity for $\Phi$. But it is too complicated. We may further simplify  the above formula at the point $x=\hat{x}$.
\begin{eqnarray*}
f''(\hat{x}) &=&\lambda d(\beta\gamma-1)\frac{(\beta \hat{x}+1)^{d-2}}{(\hat{x}+\gamma)^{d+2}}\left((d-1)\beta(\hat{x}+\gamma)-(d+1)(\beta \hat{x}+1)\right) \\
&=&\frac{1}{(\beta \hat{x}+1)(\hat{x}+\gamma)}\left((d+1)(\beta \hat{x}+1)-(d-1)\beta(\hat{x}+\gamma)\right)\\
&=& \frac{d+1}{\hat{x}+\gamma}-\frac{(d-1)\beta}{\beta \hat{x}+1}\\
&=& \frac{d(1-\beta \gamma)}{(\beta \hat{x}+1)(\hat{x}+\gamma)}+ \frac{1}{\hat{x}+\gamma}+\frac{\beta}{\beta \hat{x}+1}\\
&=& \frac{1}{\hat{x}}+ \frac{1}{\hat{x}+\gamma}+\frac{\beta}{\beta \hat{x}+1}.
 \end{eqnarray*}

So we have
\begin{align}
\left(\ln\left(\Phi(\hat{x})\right)\right)' = -\frac{f''(\hat{x})}{2} =-\frac{1}{2}(\frac{1}{\hat{x}}+ \frac{1}{\hat{x}+\gamma}+\frac{\beta}{\beta \hat{x}+1}).\label{eq:heuristic}
\end{align}
We apply the heuristics and assume that \eqref{eq:heuristic} simply holds for all $x$. This gives us a differential equation
  \[\left(\ln(\Phi(x))\right)'  =-\frac{1}{2}(\frac{1}{x}+ \frac{1}{x+\gamma}+\frac{\beta}{\beta x+1}).\]
The solution of this differential equation is
\[\ln(\Phi(x))=-\frac{1}{2} \ln (x (x+\gamma) (\beta x+1))+C_1,\]
which gives that
\[\Phi(x)=\frac{C_2}{\sqrt{x (\beta x +1)(x+\gamma)}},\]
where $C_1,C_2$ are some constants. This gives the potential function we used in the paper.

There are also other equations which hold for the fixed point $\hat{x}$. Choosing which equation for $\hat{x}$ to heuristically extend to all $x$ may affect the potential function we obtained. For example, \eqref{eq:heuristic-identity} implies that
\[\frac{1}{d \hat{x}}=\frac{1-\beta\gamma}{(\beta\hat{x}+1)(\hat{x}+\gamma)}=\frac{1}{\hat{x}+\gamma}-\frac{\beta}{\beta \hat{x}+1},\]
thus we can rewrite $f''(\hat{x})$ as
\[f''(\hat{x})=\frac{1}{\hat{x}}+ \frac{1}{\hat{x}+\gamma}+\frac{\beta}{\beta \hat{x}+1}= \frac{d+1}{d \hat{x}}+\frac{2\beta}{\beta \hat{x}+1}.\]
This gives that
\[ \left(\ln(\Phi(\hat{x}))\right)'= -\frac{d+1}{2d \hat{x}}-\frac{\beta}{\beta \hat{x}+1}.\]
We can also treat this equation as a differential equation for variable $x$ and solving it gives us
\[\Phi(x)=\frac{c}{x^{\frac{d+1}{2d}} (\beta x +1)},\]
which is the potential function used in~\cite{LLY}.

\end{document}